\documentclass[a4paper]{amsart}

\usepackage{amsmath,amssymb,amsthm,amscd}
\usepackage{latexsym}
\usepackage{listings}

\newtheorem{lemma}{Lemma}
\newtheorem{theorem}{Theorem}
\usepackage{tikz}
\usepackage{pgfplots}
\usepackage{mathtools}
\usepackage{subfigure}
\usepackage{algpseudocode}
\usepackage{algorithm}
\usepackage{siunitx}
\usepackage{mathtools}
\usepackage{multirow}
\usepackage{hyperref}

\definecolor{darkgreen}{rgb}{0,0.5,0}

\newcommand{\Asym}{A_{\mbox{\footnotesize sym}}}
\newcommand{\bsym}{\vec{b}_{\mbox{\footnotesize sym}}}
\renewcommand{\vec}[1]{{\bf #1}}
\newcommand{\sgn}{\text{sgn}}
\DeclareMathOperator{\spn}{span}

\usepackage{xcolor}
\definecolor{newcolor}{rgb}{.8,.349,.1}

\renewcommand{\vec}[1]{{\bf #1}}

\calclayout

\pgfplotsset{compat=1.18}

\begin{document}

\title[Rational Krylov for decoding edge-based compressed images]{An extended Krylov subspace method for decoding edge-based compressed images by homogeneous diffusion}

\author{Volker Grimm}
\address{Karlsruhe Institute of Technology (KIT), Institute for Applied and Numerical Mathematics, D-76131 Karlsruhe, Germany}
\email{Volker.Grimm$@$kit.edu}

\author{Kevin Liang}

\subjclass[2020]{Primary, 65F60, secondary 94A08}

\keywords{Rational Krylov subspace method, multigrid method, inpainting, time integration}

\date{\today}

\begin{abstract}
The heat equation is often used in order to inpaint dropped data in inpainting-based lossy compression schemes. We propose an alternative way to numerically solve the heat equation by an extended Krylov subspace method. The method is very efficient with respect to the computation of the solution of the heat equation at large times. And this is exactly what is needed for decoding edge-compressed pictures by homogeneous diffusion.
\end{abstract}

\maketitle

\section{Introduction}\label{intro}
Inpainting-based compression of images refers to the idea to identify prominent data in an image and to only store this data. All other data is disregarded and, when needed, reconstructed by inpainting.
In particular, we will consider edge-based compressed images, where the edges of an image together with adjacent grey/colour data are stored (cf.~\cite{Carlsson88,Elder99,HuMo89,ZeRot86}).
This idea can be seen as a second-generation image coding method where the properties of the human visual system are taken into account (cf.~\cite{Reidetal97}). The edge-based compression works very well for cartoon-like images (cf.~\cite{Mainbergeretal11}).
In order to improve the quality of the reconstruction for natural images, we also compress  images based on dithering. Dithering also works due to the perception of images by the human visual cortex.
We will work with these two basic coding techniques. But since our new contribution refers to the decoding, more advanced coding techniques (cf.~\cite{Hoeltgenetal13,Mainbergeretal12}) can easily be combined with our approach. We would also like to mention that, while our proposal deals with homogeneous diffusion, the proposed method might be carried over to
more general linear evolution equations, e.g. in image registration (e.g.~\cite{Grietal06}),
and to nonlinear partial differential equations used for inpainting by the help of exponential integrators, in which the linear part is solved as proposed in this work. More information on exponential integrators might be found in the survey by Hochbruck and Ostermann \cite{hoacta10} and information on advanced image inpainting methods by partial differential equations can be found in Schoenlieb's book \cite{inpaintingcarola15}.

In order to review the basic idea of inpainting-based compression of images, let $f:\Omega \rightarrow \mathbb{R}$ be a given grey-scale picture. $f(\vec{x})$ refers to the
intensity of light and $\Omega$ is the rectangular domain of the picture. In a colour picture, any channel is treated in the same way.
After the compression of the picture, the intensities are only known on a subdomain $K \subset \Omega$. This splits the image in a know part $K$ and an unknown part $\Omega \backslash K$. To flag the stored pixels in an efficient way, we will use the function
\[
  c(\vec{x})=
  \left\{
    \begin{array}{ll}
     1 & \mbox{for}~\vec{x} \in K, \\
     0 & \mbox{for}~\vec{x} \in \Omega \backslash K,
    \end{array}
  \right.
\]
which we will refer to as \emph{inpainting mask}. With the help of the inpainting mask, the compressed image $f_c$ can be written as $f_c=cf$. In the middle of Figure~\ref{overallidea}, such a compressed picture of the picture on the left-hand side is shown. Only the pixels that are not black are stored. This data is sufficient for the reconstruction on the right-hand side of Figure~\ref{overallidea}.
\begin{figure}[t]
        \centering
        \subfigure[original] {\includegraphics[height=5cm]{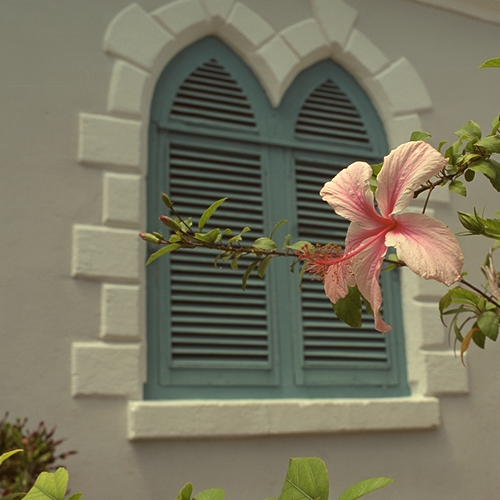}} \hfill
        \subfigure[compression] {\includegraphics[height=5cm]{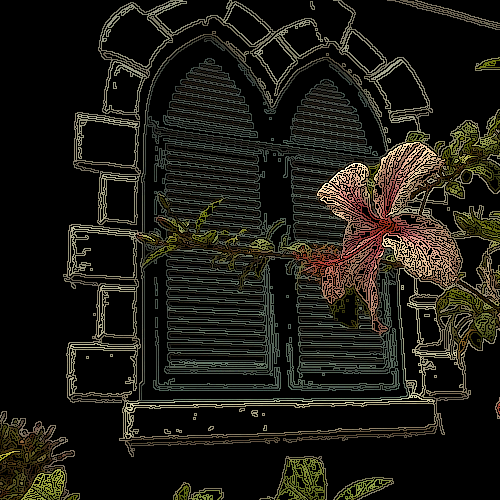}} \hfill
        \subfigure[reconstruction] {\includegraphics[height=5cm]{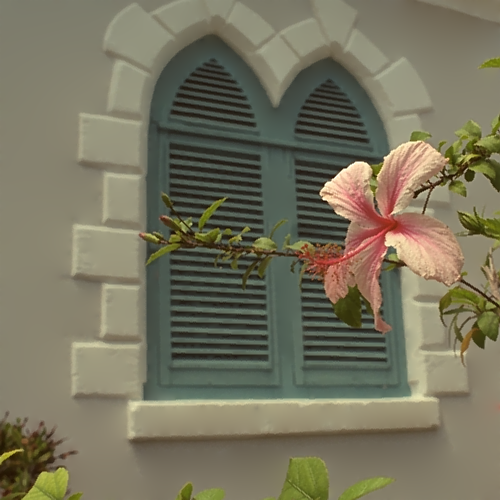}}
        \caption{Sketch of the compression scheme: in the encoding step the original picture is reduced to a subset, in the decoding phase the  original is reconstructed from this subset.} \label{overallidea}
\end{figure}
For the reconstruction, we inpaint the missing data by the heat equation. The system reads
\begin{align*}
        \partial_t u(\vec{x},t) &=(1-c(\vec{x}))\Delta u(\vec{x},t), && \mbox{in}~\Omega \times [0,\infty), \\
        \partial_n u(\vec{x},t)&=0, && \mbox{in}~\partial \Omega \times [0,\infty),
\end{align*}
with Neumann boundary conditions and the compressed image
\[
        u(\vec{x},0)=c(\vec{x})f(\vec{x}), \qquad \vec{x} \in \Omega,
\]
as initial data. The reconstructed picture is the solution $u(\vec{x},t)$ of the above PDE at a large time $t$.
The new contribution is an efficient method to solve the discretised heat equation at a prescribed time $t$, directly. The discretisation of the heat equation leads to a huge system of ordinary differential equations
\begin{equation} \label{ode}
  \vec{y}' = A\vec{y}, \qquad \vec{y}(0)=\vec{b}\,,
\end{equation}
where $\vec{b}$ are the pixels of the compressed image written as a vector. The reconstructed image $\vec{y}(t)$ at time $t$ is given by the matrix exponential times the vector $\vec{b}$, i.e. $\vec{y}(t)=e^{tA}\vec{b}$.

Recently, rational Krylov subspace methods have been found to be an excellent choice for the approximation of the matrix exponential, that is, the solution of \eqref{ode} at time $t$.
Rational Krylov subspaces have been considered first by Axel Ruhe (cf. \cite{Ruhe84,Ruhe98}). They also turned out to be useful in inverse problems (e.g.~\cite{BreNoRe12,BuDoRei17,invkry,RamRei19}).
If $A$ is an operator or an arbitrarily large matrix with a field-of-values in the left complex half-plane, the matrix exponential times a vector can be approximated reliably for an arbitrary time $t>0$ under reasonable assumptions on the vector (cf. \cite{GG13,GG14,ratkryphi11,GGautosmooth17}).
If the matrix $A$ is symmetric and the field-of-values is on the negative real axis, then rational Krylov methods are known that converge fast without any restrictions on the vector (cf. \cite{And81,marlis_jasper}). For finite subintervals of the negative real line, even super-linear convergence is obtained (cf. \cite{beckermann_guettel12}).

Our matrix $A$ is not symmetric, but nevertheless allows for the use of a well-chosen rational Krylov subspace such that a fast convergence is obtained.
We will approximate the solution of system \eqref{ode} in extended Krylov subspaces of the form
\begin{equation} \label{exKryIntro}
 \mathcal{E}_{m-1}((\gamma I-A)^{-1},\vec{b}) = \spn \{ \vec{b}, A\vec{b}, (\gamma I-A)^{-1}\vec{b}, \cdots, (\gamma I-A)^{-m+2}\vec{b}\}\,,
 \nonumber
\end{equation}
where $\gamma > 0$ (cf. \cite{Druskin_Knizhnerman98,GG13,KnizhSimoncini09}).
After the computation of an orthonormal basis $V_m \in \mathbb{R}^{n \times m}$ of this space and the compression $S_m=V_m^TAV_m$ of the huge matrix $A$ to a small $m \times m$ matrix, the \emph{Krylov approximation} $\vec{f}_m$ is given by
\begin{equation} \label{kryapprox}
        \vec{f}_m = \|\vec{b}\|V_me^{tS_m}\vec{e}_1, \qquad
        \vec{e}_1= \begin{pmatrix} 1, 0, \cdots, 0 \end{pmatrix}^T \in \mathbb{R}^m\,.
\end{equation}
This way, the solution of the huge system \eqref{ode} is reduced to the solution of a small system of size $m \times m$, that is, to the computation of $e^{tS_m}\vec{e}_1$. For this purpose, methods for small matrices can be used (cf. \cite{AlMohyHigham11,highambook}).
It will turn out, that a small $m$ is sufficient for arbitrary large matrices $A$.
The choice of the subspace that leads to this favourable property is intricate due to two restrictions. A good error estimate is necessary in order to estimate the accuracy of the method and the computation of the vectors $(\gamma I-A)^{-m}\vec{b}$ requires the efficient solution of linear systems. This can be done by multigrid methods.

The paper is organised as follows: After the introduction in this section, the encoding of pictures is briefly discussed in Section~\ref{sec:encoding}. In Section~\ref{sec:discretization}, the discretisation of the heat equation is described. The new decoding scheme and the  main result that it works independent of the size of the picture for a given large time $t$ is shown in Section~\ref{sec:decoding}. The multigrid method adapted to our purposes as an efficient method to solve the linear systems is discussed in Section~\ref{sec:multigrid}. Numerical experiments with the decoding scheme as illustration of our method are conducted in Section~\ref{sec:numex}. Here and everywhere else, we will use pictures of the Kodak lossless colour image suite (cf. \cite{Kodak}). The work closes with a brief conclusion as Section~\ref{sec:conclusion}.

\section{Encoding} \label{sec:encoding}
In this section, we briefly describe the encoding. The basic idea is to determine a binary mask $\vec{c}$ of the same size as the picture that indicates which of the grey/colour values are stored. The choice of this mask determines the compression. The fewer pixels we have to store the higher will be the compression. The choice of the mask is also important for the obtainable quality of the reconstructed image. We consider two basic methods in order to determine a good mask for a later inpainting of the picture. For simplicity, we will not consider the generation of more elaborate masks by advanced coding techniques (cf.~\cite{Hoeltgenetal13,Mainbergeretal12}). Our new decoding algorithm works with any mask $\vec{c}$. For the demonstration of our approach, the two basic masks will suffice.
\subsection{Edge detection}
Edges are very important for the perception of images by the human brain (cf. \cite{Marr76}). Therefore, one often starts with detecting edge information in an image. One classic and often used idea is to identify edges as zero-crossings of the Laplacian of an image that has been smoothed by a Gaussian filter, the Marr-Hildreth edge detector (cf.~\cite{MarrHildreth80}). For a colour picture, $\vec{f}=(f_1,f_2,f_3)^T$, the Laplacian is defined as the sum of the Laplacians over all channels:
\[
  \Delta \vec{f} = \sum_{k=1}^3 \Delta f_k\,.
\]
So, the idea is to store the pixels with large absolute value of this Laplacian.
In order to remove zero-crossings that arise from small oscillations in the image, the magnitude of the gradient $\nabla \vec{u}$ of the image at every pixel is used in addition. All edges are removed where the gradient is below a certain threshold. This is basically the idea of the Canny edge detector (cf. \cite{Canny86}).
\subsection{Dithering}
If one uses the above idea for natural images, the highly textured parts of the image are strongly emphasised in contrast to the background. This can be seen in Figure~\ref{laplacian}.
The background seems to be too blurred.
In order to improve the representation of smoother regions, it is proposed to choose the edge data proportional to the absolute value of the Laplacian (cf.~\cite{Belhachmietal09}). In order to follow this suggestion, we use Floyd--Steinberg dithering, cf.~\cite{FloStein76}, for the modulus of the Laplacian.
This method also allows for a simple method to prescribe the percentage of the pixels to be stored. For example, if one wishes to store $10\%$ of the pixels, the largest modulus of the Laplacian in the picture is scaled such that the average corresponds to a tenth of the value of a white pixel. If the maximal value of a white pixel is $255$, the average corresponds to $0.1 \times 255 = 25.5$. In the course of the Floyd--Steinberg dithering, about 10\% of the pixels will automatically be stored.
\begin{figure}[t]
        \centering
        \subfigure[original] {\includegraphics[height=5cm]{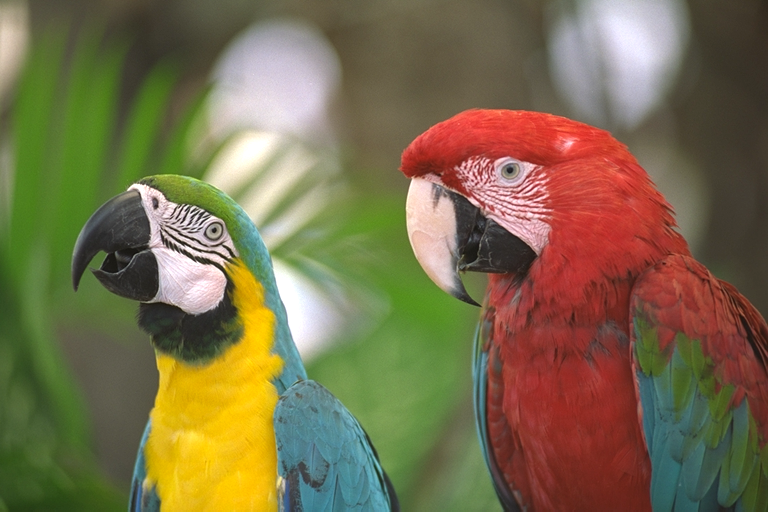}} \quad
        \subfigure[modulus Laplacian] {\includegraphics[height=5cm]{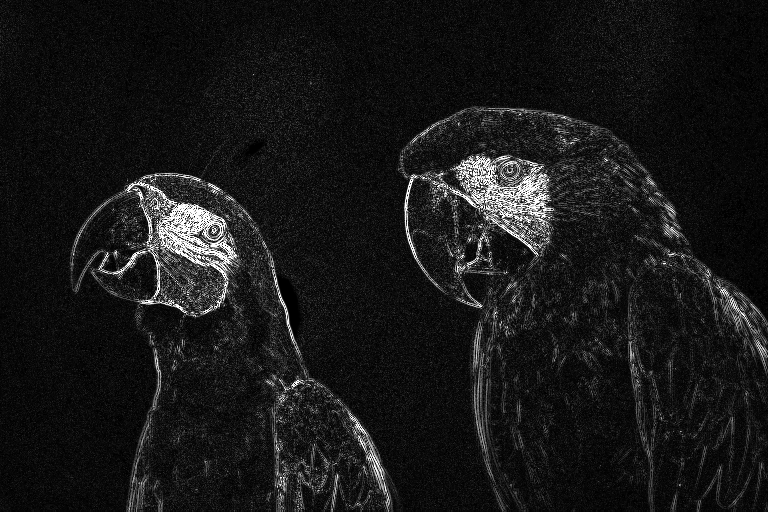}}\\
        \subfigure[threshold (10\%)] {\includegraphics[height=5cm]{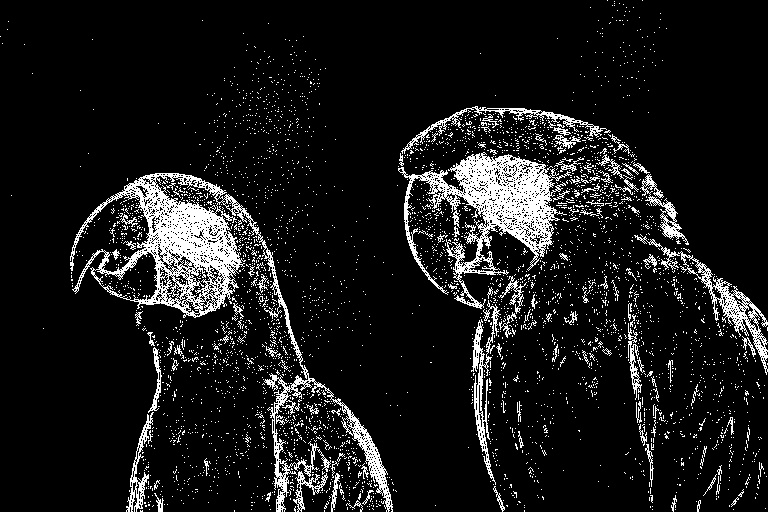}} \quad
        \subfigure[reconstruction] {\includegraphics[height=5cm]{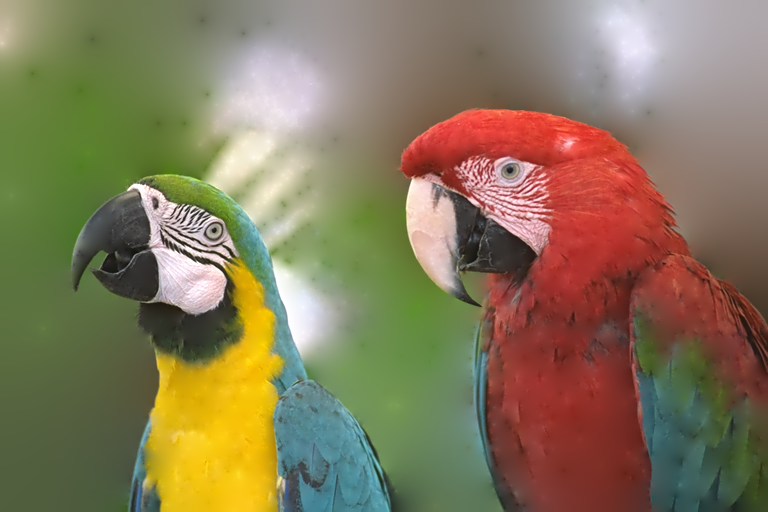}}
        \caption{The mask has been chosen such that 10\% of the pixels with the largest modulus of the Laplacian are contained.}
        \label{laplacian}
\end{figure}
\begin{figure}[t]
        \centering
        \subfigure[dithering (10\%)] {\includegraphics[height=5cm]{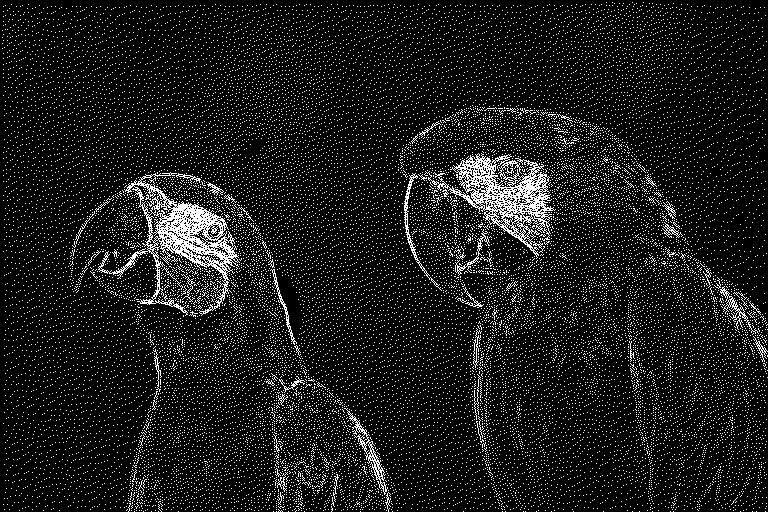}} \quad
        \subfigure[reconstruction] {\includegraphics[height=5cm]{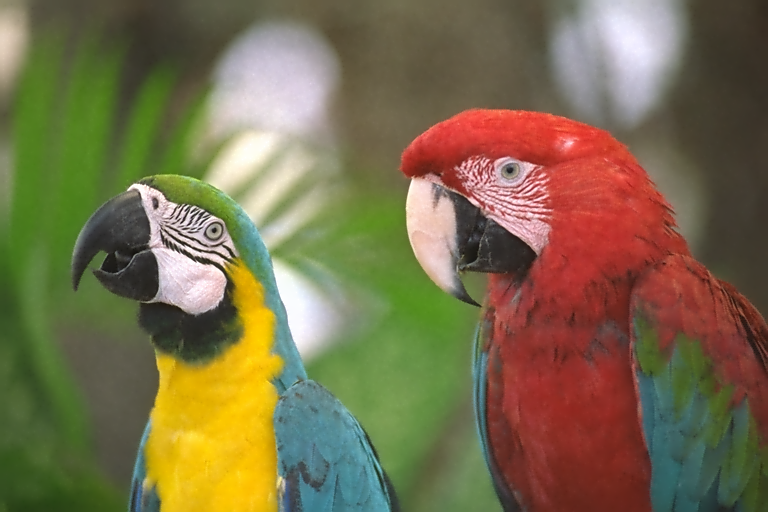}}
        \caption{The mask has been chosen such that the modulus of the Laplacian has been scaled to $10\%$ and dithering has been applied.}
        \label{dithermask}
\end{figure}

In the following, we will refer to the first method as \emph{edge-based} compression and the second one as \emph{dithering-based} compression. Dithering improves the display of the background in the reconstructed image as can be seen in Figure~\ref{dithermask}. The edge-based compression is superior for cartoon-like images, vector graphics, pictograms, and letters, where the edges are the most important image feature. The dithering might be superior for natural images, as shown above.
\subsection{Further compression}
If one saves all important edges, one stores more pixels as with the dithering approach. But this is mitigated by the fact that a subsampling of the edge-information is possible (cf. \cite{Mainbergeretal11}). The two-dimensional signal is transformed to a one-dimensional one. Afterwards, since the one-dimensional signal shows more continuity, only every $d$-th value is stored.
The masks generated by the edge-based approach are better suited for subsampling than the ones generated by dithering. A presmoothing to the one-dimensional signals might be applied in order to improve the subsampling even more.

The subsampled data is further compressed by uniform quantisation and then by applying an entropy coder.
Since we mostly use the programming language {\tt Python} in our numerical experiments, we applied the function {\tt savez\_compressed} of Python's NumPy library as the entropy coder.


\section{Discretisation} \label{sec:discretization}

For decoding the compressed pictures, the $N_x \times N_y$ pixels of the original picture are interpreted as a finite-difference approximation to the heat equation. The mask $\vec{c}$ turns into a binary mask, where $1$ indicates that the pixel has been stored. Every such pixel is treated as a discretised Dirichlet boundary condition with the pixel value as the boundary value. Any boundary pixel, which is not a stored pixel, is treated as a homogeneous Neumann boundary. The Laplacian is discretised by the standard stencil
\begin{equation}
\Delta^{\vec{h}} \mathrel{\hat =}
\begin{bmatrix}
0 & \frac 1 {h_x^2} & 0\\
\frac 1 {h_y^2} & -\left(\frac 2 {h_x^2} + \frac 2 {h_y^2}\right) & \frac 1 {h_y^2}\\
0 & \frac 1 {h_x^2} & 0
\end{bmatrix}
\end{equation}
with grid constants $h_x,h_y$ in $x,y-$direction, respectively. The stencil is applied at any pixel which has not been stored.
As usual in image processing, we will assume that the grid constants are one in both directions on the finest grid, which corresponds to the original picture.
We illustrate the discretisation by the small example in Figure~\ref{example3times2}. The grey pixels indicate the stored pixels, that is, the Dirichlet boundary data, where the mask $\vec{c}$ is one. The discretised heat equation reads
\begin{equation} \label{heatdisc}
  \vec{y}'=A\vec{y}, \qquad \vec{y}(0)=\vec{b},
\end{equation}
where $A$ and $\vec{b}$ are as follows,
\[
A=\begin{bmatrix}
-2 &  1 & 0 &  1 & 0 &  0 \\
 1 & -3 & 1 &  0 & 1 &  0 \\
 0 &  0 & 0 &  0 & 0 &  0 \\
 1 &  0 & 0 & -2 & 1 &  0 \\
 0 &  0 & 0 &  0 & 0 &  0 \\
 0 &  0 & 1 &  0 & 1 & -2
\end{bmatrix}\,,
\qquad
\vec{b}=\begin{bmatrix}
0 \\
0 \\
f_{3}^h \\
0 \\
f_{5}^h \\
0 \\
\end{bmatrix}.
\]
The matrix $R$,
\[
R =
\begin{bmatrix}
 1 & 0 & 0 & 0 & 0 & 0 \\
 0 & 1 & 0 & 0 & 0 & 0 \\
 0 & 0 & 0 & 1 & 0 & 0 \\
 0 & 0 & 0 & 0 & 0 & 1
\end{bmatrix} \,,
\]
selects, by multiplication from the left-hand side, the rows of the matrix $A$ that correspond to inner pixels where the Laplacian stencil is applied. $R^T$ blows a vector that corresponds to the inner pixels up to the full size of our picture while setting the boundary pixels to zero. The projector $P=R^TR$ projects to the orthogonal complement of the space spanned by $\vec{b}$. As a consequence $PA=A$.
With the help of the matrices $R$ and $R^T$, respectively, a symmetric matrix $\Asym$ can be extracted from the matrix $A$ as well as a reduced vector $\bsym$,
\[
\Asym=RAR^T=
\begin{bmatrix}
 -2 &  1 &  1 &  0 \\
  1 & -3 &  0 &  0 \\
  1 &  0 & -2 &  0 \\
  0 &  0 &  0 & -2
\end{bmatrix},
\qquad
\bsym=RA\vec{b}\,,
\]
that allow for an alternative representation of the solution of system \eqref{heatdisc} given in Lemma~\ref{exactsolrep}.
\begin{figure}
\caption{Example of image with dimension $3 \times 2$ and its row-major numbering} \label{example3times2}

\hspace{0.25cm}

\begin{center}
\begin{tikzpicture}[scale=1.0]
\filldraw[lightgray] (1,0) rectangle (2,1);
\filldraw[lightgray] (2,1) rectangle (3,2);
\draw[thick] (0,0) rectangle (3,2);
\draw[thick] (0,1) -- (3,1);
\foreach \x in {1,2}
   \draw[thick] (\x,0) -- (\x,2);
\node at (0.5,0.5) {$u_4^h$};
\node at (1.5,1.5) {$u_2^h$};
\node at (0.5,1.5) {$u_1^h$};
\node at (1.5,0.5) {$f_5^h$};
\node at (2.5,1.5) {$f_3^h$};
\node at (2.5,0.5) {$u_6^h$};
\end{tikzpicture}
\hspace{0.5cm}
\begin{tikzpicture}[scale=1.0]
\draw[thick] (0,0) rectangle (3,2);
\draw[thick] (0,1) -- (3,1);
\foreach \x in {1,2}
   \draw[thick] (\x,0) -- (\x,2);
\node at (0.5,0.5) {$4$};
\node at (1.5,1.5) {$2$};
\node at (0.5,1.5) {$1$};
\node at (1.5,0.5) {$5$};
\node at (2.5,1.5) {$3$};
\node at (2.5,0.5) {$6$};
\end{tikzpicture}
\end{center}
\end{figure}
\begin{lemma} \label{exactsolrep}
The exact solution of equation \eqref{heatdisc} can be written as
\begin{equation} \label{altsolrep}
   \vec{y}(t) = e^{t A}\vec{b} = \vec{b} + R^T \left( t \varphi_1(t\Asym)\bsym \right)\,,
\end{equation}
where $\varphi_1(z)=(e^z-1)\slash z$.
\end{lemma}
\begin{proof}
It is well known that the matrix exponential solves the ordinary differential equation \eqref{heatdisc}. Hence, from here, we obtain
\[
  \vec{y}(t)=e^{tA}\vec{b}=\vec{b}+t\varphi_1(tA)A\vec{b} = \vec{b}+t\varphi_1(tPA)PA\vec{b}
  = \vec{b} + R^T \left( t\varphi_1(t\Asym) \bsym\right)\,,
\]
by $PA=A$ and the fact that
\[
  \varphi_1(R^TRA)R^T =R^T\varphi_1(RAR^T)
\]
according to Corollary~1.34 on page $21$ in Higham's book (cf.~\cite{highambook}).
\end{proof}
In the alternative representation \eqref{altsolrep} of the solution, one can see that the stored pixels in $\vec{b}$ are never altered due to the properties of $R^T$, which is also true for the exact solution of \eqref{heatdisc}, of course.
It follows by the Gershgorin disk theorem that
$\Asym$ has only negative eigenvalues and hence is invertible.
(Strictly speaking, $\Asym$ has only negative eigenvalues as soon as at least one boundary pixel exists in the interior of the rectangular domain of the picture.)
From the stencil, one can easily read off that the matrix $\Asym$ is symmetric.
Based on these facts, the following theorem shows that $(\gamma I - A)$ is invertible for all $\gamma > 0$, which is crucial for our decoding method.
\begin{lemma}\label{gImArep} For $\gamma > 0$, we have
\[
(\gamma I - A)^{-1} = \tfrac{1}{\gamma}I +\tfrac{1}{\gamma}R^T(\gamma I-\Asym)^{-1}RA\,.
\]
\end{lemma}
\begin{proof} We compute
\begin{align*}
(\gamma I-&A)
  \left(
   \tfrac{1}{\gamma}I +\tfrac{1}{\gamma}R^T(\gamma I-\Asym)^{-1}RA
  \right)
 =
 \tfrac{1}{\gamma}(\gamma I-A)+\tfrac{1}{\gamma}(\gamma I-A)R^T(\gamma I-\Asym)^{-1}RA \\
 &=
 I-\tfrac{1}{\gamma}A +\tfrac{1}{\gamma} P (\gamma I-A)R^T(\gamma I-\Asym)^{-1}RA
 =
 I-\tfrac{1}{\gamma}A + \tfrac{1}{\gamma} R^T(\gamma I-\Asym)(\gamma I-\Asym)^{-1}RA \\
 &=
 I-\tfrac{1}{\gamma}A+\tfrac{1}{\gamma}PA=I\,.
\end{align*}
\end{proof}

\section{Decoding by the extended Krylov subspace method} \label{sec:decoding}

Extended Krylov subspaces for invertible matrices use the matrix as well as its inverse (cf.~\cite{Druskin_Knizhnerman98,GG13,KnizhSimoncini09}). Let $C$ be an invertible matrix and $\vec{b}$ a vector of a suitable dimension. Then the extended Krylov subspace $\mathcal{E}_p^q(C,\vec{b})$ is defined as
\[
\mathcal{E}_p^q (C,\vec{b}) = \spn \left\{ \vec{b},C^{-1}\vec{b},\ldots,C^{-q}\vec{b},C\vec{b},\ldots,C^{p-1}\vec{b}\right\}\,.
\]
Due to Lemma~\ref{gImArep}, $\gamma I-A$ is invertible and hence we can set $C=(\gamma I-A)^{-1}$. We will use extended Krylov subspaces where $q$ is always one of the following form
\[
    \mathcal{E}_m((\gamma I-A)^{-1},\vec{b})
    =\mathcal{E}_{m-1}^1((\gamma I-A)^{-1},\vec{b})
    =\spn \{ \vec{b}, A\vec{b}, (\gamma I-A)^{-1}\vec{b}, \cdots, (\gamma I-A)^{-m+2}\vec{b}\}\,.
\]
Note, that we have used that $\spn \{\vec{b}, (\gamma I-A)\vec{b} \}=\spn \{\vec{b}, A\vec{b}\}$, here.
We start by computing an orthonormal basis $V_m \in \mathbb{R}^{n \times m}$ of the extended Krylov subspace by Algorithm~\ref{alg:arnlike}. Then, we compute the compression $S_m=V_m^TAV_m \in \mathbb{R}^{m \times m}$ of the large matrix~$A$, and finally the \emph{Krylov approximation} $\vec{f}_m$ as given in \eqref{kryapprox}.
In order to understand the properties of the Krylov algorithm, we study the algorithm via the symmetric matrix $\Asym$, the initial vector $\bsym$, and the alternative solution representation \eqref{altsolrep}.

\begin{algorithm}
        \caption{Symmetric Arnoldi-like algorithm}
\label{alg:arnlike}
        \begin{algorithmic}
                \State {Set $\vec{v}_1=\vec{b}\slash \|\vec{b}\|$, $\vec{v}_0=\vec{0}$}
\For {m=1,2,3,\ldots}
\If {$m=1$}
\State {$\vec{u}=A\vec{v}_1$}
\Else
\State{$\vec{u}=(\gamma I-A)^{-1}\vec{v}_m$}
\EndIf
\For {$j=m-1,m$}
\State {$h_{jm}=(\vec{u},\vec{v}_j)$}
\EndFor
\State {$\tilde{\vec{v}}_{m+1}=\vec{u}-h_{mm}\vec{v}_m-h_{m-1,m}\vec{v}_{m-1}$}
\State {$\vec{v}_{m+1}=\tilde{\vec{v}}_{m+1}\slash h_{m+1,m}$}, {$h_{m+1,m}=\|\tilde{\vec{v}}_{m+1}\|$}
\EndFor
\end{algorithmic}
\end{algorithm}
\begin{algorithm}[ht]
\caption{Symmetric Arnoldi} \label{alg:arn}
\begin{algorithmic}
\State {Set $\vec{w}_1=\bsym\slash \|\bsym\|$, $\vec{w}_0=\vec{0}$}
\For {m=1,2,3,\ldots}
\State{$\vec{u}=(\gamma I-\Asym)^{-1}\vec{w}_m$}
\For {$j=m-1,m$}
        \State {$\widehat{h}_{jm}=(\vec{u},\vec{w}_j)$}
\EndFor
        \State {$\tilde{\vec{w}}_{m+1}=\vec{u}-\widehat{h}_{mm}\vec{w}_m-\widehat{h}_{m-1,m}\vec{w}_{m-1}$}
        \State {$\vec{w}_{m+1}=\tilde{\vec{w}}_{m+1}\slash \widehat{h}_{m+1,m}$}, {$\widehat{h}_{m+1,m}=\|\tilde{\vec{w}}_{m+1}\|$}
\EndFor
\end{algorithmic}
\end{algorithm}
\noindent The computation of the basis $V_m$ by Algorithm~\ref{alg:arnlike} can be compared with the symmetric Arnoldi method for a related Krylov subspace as outlined in Lemma~\ref{lembasisVWext}.
\begin{lemma} \label{lembasisVWext}
Let $W_{m-1}=[\vec{w}_1,\cdots,\vec{w}_{m-1}]$ be the Arnoldi basis for the Krylov subspace $\mathcal{K}_{m-1}((\gamma I-\Asym)^{-1},\bsym)$ according to Algorithm~\ref{alg:arn}. Then, the orthonormal basis for the extended Krylov subspace
\[
  \mathcal{E}_m ((\gamma I-A)^{-1},\vec{b})
  = \spn \{ \vec{b}, A\vec{b}, (\gamma I-A)^{-1}\vec{b}, \cdots, (\gamma I-A)^{-m+2}\vec{b}\}
\]
according to Algorithm~\ref{alg:arnlike} reads
\[
        V_m=\left[ \vec{v}_1,\ldots, \vec{v}_m \right]=\left[ \vec{v}_1, R^TW_{m-1}\right]\,, \qquad \vec{v}_1=\frac{1}{\|\vec{b}\|}\vec{b}\,.
\]
\end{lemma}
\begin{proof}
        The idea of the following proof is to compare the Arnoldi-like algorithm with the Arnoldi algorithm (cf. Algorithm~6.1 in Saad's book \cite{Saaditer}) for the standard Krylov space $\mathcal{K}_{m-1}((\gamma I-\Asym)^{-1},\bsym)$. The Arnoldi-like algorithm and the Arnoldi algorithm are the same as Algorithm~\ref{alg:arnlike} and Algorithm~\ref{alg:arn}, respectively, with the difference that the for-loop $j$ runs from $j=0,\ldots,m$.
In the Arnoldi-like algorithm, we compute $h_{jm}$ for $j=1,\ldots,m$, which guarantees that the generated vectors are perpendicular to each other.
The Arnoldi-like algorithm leads to the following. Let $\vec{v}_1=\frac{1}{\|\vec{b}\|}\vec{b}$, which is the obvious start. Then
\begin{align*}
  A\vec{v}_1 &= \frac{1}{\|\vec{b}\|} A\vec{b}  = \frac{1}{\|\vec{b}\|}R^T\bsym, \\
        h_{11} &= \vec{v}_1^T A\vec{v}_1 =\frac{1}{\|\vec{b}\|^2} \vec{b}^TA\vec{b}=\frac{1}{\|\vec{b}\|^2} \vec{b}^TPA\vec{b}=0,\\
  \tilde{\vec{v}}_2 &= A\vec{v}_1-h_{1,1}\vec{v}_1 = \frac{1}{\|\vec{b}\|} R^T \bsym, \\
        h_{21} &=\|\tilde{\vec{v}}_2\|=\frac{\|\bsym\|}{\|\vec{b}\|}, \\
        \vec{v}_2 &= \frac{1}{h_{21}}\tilde{\vec{v}}_2=R^T \frac{1}{\|\bsym\|} \bsym = R^T\vec{w}_1\,.
\end{align*}
Our statement is proved for $m=2$. For $m\geq 3$, we have
\begin{align*}
 h_{1m} &= \vec{v}_1^T(\gamma I-A)^{-1}\vec{v}_m
         = \vec{v}_1^T
           \left(
             \tfrac{1}{\gamma}I+\tfrac{1}{\gamma}R^T(\gamma I-\Asym)^{-1}RA
           \right)\vec{v}_m
         = \tfrac{1}{\gamma}\vec{v}_1^T\vec{v}_m
           + \tfrac{1}{\gamma}(R\vec{v}_1)^T(\gamma I-\Asym)^{-1}RA\vec{v}_m=0\,,
\end{align*}
since $\vec{v}_1^T\vec{v}_m=0$ and $R\vec{v}_1=\vec{0}$.
For $j=2,\ldots,m$, we obtain
\begin{align*}
 h_{jm} &= \vec{v}_j^T(\gamma I-A)^{-1}\vec{v}_m
         = \vec{v}_j^T
           \left(
             \tfrac{1}{\gamma}I+\tfrac{1}{\gamma}R^T(\gamma I-\Asym)^{-1}RA
           \right)\vec{v}_m
         = \tfrac{1}{\gamma}\vec{v}_j^T\vec{v}_m
           + \tfrac{1}{\gamma} (R\vec{v}_j)^T(\gamma I-\Asym)^{-1}RA\vec{v}_m\\
         &= \tfrac{1}{\gamma}(R^T\vec{w}_{j-1})^TR^T\vec{w}_{m-1}
           + \tfrac{1}{\gamma} (RR^T\vec{w}_{j-1})^T(\gamma I-\Asym)^{-1}RAR^T\vec{w}_{m-1}\\
	 &= \tfrac{1}{\gamma}\vec{w}_{j-1}^T\vec{w}_{m-1}
           + \tfrac{1}{\gamma} \vec{w}_{j-1}^T(\gamma I-\Asym)^{-1}\Asym\vec{w}_{m-1} 
\end{align*}
With the help of the relation
\begin{equation} \label{resrel}
        (\gamma I-\Asym)^{-1}\Asym = -I + \gamma(\gamma I-\Asym)^{-1}\,,
\end{equation}
we conclude
\begin{align*}
 h_{jm} &= \tfrac{1}{\gamma}\vec{w}_{j-1}^T\vec{w}_{m-1}
           + \tfrac{1}{\gamma} \vec{w}_{j-1}^T
             \left(
                -I+\gamma(\gamma I-\Asym)^{-1}
             \right) \vec{w}_{m-1}
         = \vec{w}_{j-1}^T(\gamma I-\Asym)^{-1} \vec{w}_{m-1}
         = \widehat{h}_{m-1,j-1}, \quad j=2,\ldots,m.
\end{align*}
By \eqref{resrel}, we also have
\begin{align*}
(\gamma I-A)^{-1}\vec{v}_m
  &= \left(
       \tfrac{1}{\gamma}I+\tfrac{1}{\gamma}R^T(\gamma I-\Asym)^{-1}RA
     \right) \vec{v}_m
   = R^T\left(
        \tfrac{1}{\gamma}I+\tfrac{1}{\gamma}(\gamma I-\Asym)^{-1}\Asym
     \right)\vec{w}_{m-1}
        = R^T(\gamma I-\Asym)^{-1}\vec{w}_{m-1}\,.
\end{align*}
Finally, we find
\begin{align*}
   \tilde{\vec{v}}_{m+1} &= (\gamma I-A)^{-1}\vec{v}_m -
        \sum_{j=2}^m h_{jm}\vec{v}_j - h_{1m}\vec{v}_1
   = R^T(\gamma I-\Asym)^{-1}\vec{w}_{m-1} -\sum_{j=2}^m \widehat{h}_{j-1,m-1}R^T \vec{w}_{j-1} \\
     &= R^T \left( (\gamma I-\Asym)^{-1}\vec{w}_{m-1} -
         \sum_{j=1}^{m-1} \widehat{h}_{j,m-1}\vec{w}_j \right)
   = R^T \tilde{\vec{w}}_m
\end{align*}
as well as $h_{m+1,m}=\|\tilde{\vec{v}}_{m+1}\|=\|\tilde{\vec{w}}_m\|=\widehat{h}_{m,m-1}$ and
\[
        \vec{v}_{m+1}=\frac{\tilde{\vec{v}}_{m+1}}{h_{m+1,m}}=R^T \frac{\tilde{\vec{w}}_m}{\widehat{h}_{m,m-1}}=R^T\vec{w}_m\,.
\]
We have now proved our statement for the Arnoldi algorithm. Since the matrix $(\gamma I-\Asym)^{-1}$ is symmetric, the Arnoldi algorithm automatically reduces to the symmetric Arnoldi algorithm
(cf. Section~6.6 of Saad's book~\cite{Saaditer}).
Hence, $\widehat{h}_{j-1,m-1}=0$ for $j-1\leq m-3$, and we obtain $h_{jm}=0$ for $j \leq m-2$.
\end{proof}
The first $m$ steps of Algorithm~\ref{alg:arnlike} and the first $m-1$ steps of Algorithm~\ref{alg:arn} can be written compactly as
\begin{equation}\label{arn:rel}
 \left[
   A\vec{v}_1, (\gamma I-A)^{-1}\left[ \vec{v}_2, \cdots, \vec{v}_m\right]
 \right]
 =
 \left[
   \vec{v}_1,\cdots, \vec{v}_{m+1}
        \right] \underline{H_m} \qquad \mbox{or} \qquad
 (\gamma I-\Asym) \left[ \vec{w}_1,\cdots, \vec{w}_{m-1} \right]
 = \left[ \vec{w}_1,\cdots,\vec{w}_m \right] \underline{\widehat{H}_{m-1}}\,,
\end{equation}
respectively,
where $\underline{H_m}$ is of dimension $(m+1) \times m$ and $\underline{\widehat{H}_{m-1}}$ is of dimension $m \times (m-1)$. Both are unreduced upper Hessenberg matrices that contain the values computed in the algorithms and are otherwise set to zero. Hence, more exactly, both matrices are tridiagonal due to Lemma~\ref{lembasisVWext}.
With $H_m$ and $\widehat{H}_{m-1}$, we designate the matrices with dimensions $m \times m$ and $(m-1)\times(m-1)$, respectively.
We then have
\[
 \underline{H_m}=
 \begin{bmatrix}
     0   &         0         \\
  \vec{u} &  \underline{\widehat{H}_{m-1}}  \\
 \end{bmatrix}, \quad \vec{u}=h_{21}\vec{e}_1 = \frac{\|\bsym\|}{\|\vec{b}\|}\vec{e}_1, \quad \mbox{and}\quad \widehat{H}_{m-1}^T=\widehat{H}_{m-1}\,.
\]
The following lemma relates the compression $S_m=V_m^TAV_m$ of $A$ in the extended Krylov subspace to the compression $\widetilde{S}_{m-1}=W_{m-1}^T\Asym W_{m-1}$ of $\Asym$ in the Krylov space $\mathcal{K}_{m-1}((\gamma I-\Asym)^{-1},\bsym)$.
\begin{lemma} \label{Smstruct}
Let $\widetilde{S}_{m-1}=W_{m-1}^T\Asym W_{m-1}$. Then, for $S_m=V_m^TAV_m$, one can find
\[
 S_m=
 \begin{bmatrix}
     0   &         0         \\
 \vec{u} &  \tilde{S}_{m-1}  \\
 \end{bmatrix}, \quad \vec{u}=\frac{1}{\|\vec{b}\|}W_{m-1}^T\bsym = \frac{\|\bsym\|}{\|\vec{b}\|}\vec{e}_1,
\]
with $\vec{e}_1 \in \mathbb{R}^{m-1}$,
and therefore, $\gamma I-S_m$ is invertible for all $\gamma >0$.
\end{lemma}
\begin{proof}
By Lemma~\ref{lembasisVWext}, one obtains
\renewcommand{\arraystretch}{1.5}
\begin{align*}
S_m&=V_m^TAV_m
=
\begin{bmatrix}
 \frac{1}{\|\vec{b}\|}\vec{b} & R^TW_{m-1}
\end{bmatrix}^TA
\begin{bmatrix}
 \frac{1}{\|\vec{b}\|}\vec{b} & R^TW_{m-1}
\end{bmatrix}
=
\begin{bmatrix}
\frac{1}{\|\vec{b}\|}\vec{b}^T \\
  W_{m-1}^TR
\end{bmatrix}
\begin{bmatrix}
\frac{1}{\|\vec{b}\|}A\vec{b} & AR^TW_{m-1}
\end{bmatrix}\\
&=
\begin{bmatrix}
\frac{1}{\|\vec{b}\|^2}\vec{b}^TA\vec{b} & \frac{1}{\|\vec{b}\|}\vec{b}^TAR^TW_{m-1} \\
\frac{1}{\|\vec{b}\|}W_{m-1}^TRA\vec{b} & W_{m-1}^TRAR^TW_{m-1}
\end{bmatrix}\,.
\end{align*}
Since $A\vec{v}$ has zeros where $\vec{b}$ has entries and vice versa,
$\vec{b}^TA\vec{b}=0$ and $\frac{1}{\|\vec{b}\|}\vec{b}^TAR^TW_{m-1}=\vec{0}$ is a zero row vector of length $m-1$. Furthermore,
\[
\vec{u} =\frac{1}{\|\vec{b}\|}W_{m-1}^TRA\vec{b}
        =\frac{1}{\|\vec{b}\|}W_{m-1}^T\bsym
        =\frac{\|\bsym\|}{\|\vec{b}\|}W_{m-1}^T\vec{w}_1
        =\frac{\|\bsym\|}{\|\vec{b}\|}\vec{e}_1
\]
and
\[
\tilde{S}_{m-1}=W_{m-1}^TRAR^TW_{m-1}=W_{m-1}^T\Asym W_{m-1}\,.
\]
\noindent Hence
\renewcommand{\arraystretch}{1}
\[
 \gamma I-S_m =
 \begin{bmatrix}
    \gamma & 0 \\
   -\vec{u} & \gamma I-\tilde{S}_{m-1}
 \end{bmatrix}
\]
is invertible, since $\gamma >0$ and $\tilde{S}_{m-1}$ is a symmetric, negative-definite matrix, by Lemma~\ref{lem:Stildesym}, and $\gamma I-\tilde{S}_{m-1}$ therefore is a symmetric positive-definite matrix.
\end{proof}
\begin{lemma} \label{lem:Stildesym}
  $\tilde{S}_{m-1}$ is a symmetric and negative-definite matrix.
\end{lemma}
\begin{proof}
The symmetry follows directly from the symmetry of $\Asym$:
\[
 \tilde{S}_{m-1}^T
   = \left(W_{m-1}^T\Asym W_{m-1}\right)^T=W_{m-1}^T\Asym^T W_{m-1}
   = W_{m-1}^T\Asym W_{m-1} = \tilde{S}_{m-1}\,.
\]
Since $\Asym$ is symmetric and has only negative eigenvalues, the field-of-values of $\Asym$ is given as
\[
W(\Asym)=[\lambda_{\mbox{\scriptsize min}},\lambda_{\mbox{\scriptsize max}}] \subset (-\infty,0),
\]
where
\[
\lambda_{\mbox{\scriptsize min}}=\min_{\lambda \in \sigma(\Asym)} \lambda, \qquad
\lambda_{\mbox{\scriptsize max}}=\max_{\lambda \in \sigma(\Asym)} \lambda\,,
\]
and $\sigma(\Asym)$ is the set of eigenvalues. Hence, for the field-of-values of $\tilde{S}_{m-1}$,
\[
 W(\tilde{S}_{m-1})
   = \left\{ \frac{\vec{y}^H \tilde{S}_{m-1} \vec{y}}{\vec{y}^H\vec{y}} ~\big|~ 0 \neq \vec{y} \in \mathbb{C}^{n-1} \right\}
   \subset W(\Asym) = [\lambda_{\mbox{\scriptsize min}},\lambda_{\mbox{\scriptsize max}}]\,,
\]
since
\[
\frac{\vec{y}^HW_{m-1}^H\Asym W_{m-1} \vec{y}}{\vec{y}^H\vec{y}}
 = \frac{\vec{y}^HW_{m-1}^H\Asym W_{m-1} \vec{y}}{\vec{y}^HW_{m-1}^HW_{m-1}\vec{y}}
 = \frac{\vec{x}^H\Asym \vec{x}}{\vec{x}^H\vec{x}}\,, \qquad \vec{x}=W_{m-1}\vec{y}\,.
\]
Hence, since $ \sigma(\tilde{S}_{m-1}) \subset W(\tilde{S}_{m-1})$, all eigenvalues of $\tilde{S}_{m-1}$ are negative and
\[
        \frac{\vec{y}^H \tilde{S}_{m-1}\vec{y}}{\|\vec{y}\|^2} < 0\, \qquad \mbox{for all} \quad \vec{y}\neq \vec{0}\,,
\]
and therefore
\[
        \vec{y}^H \tilde{S}_{m-1}\vec{y} < 0, \qquad \mbox{for all} \quad \vec{y} \neq \vec{0}\,,
\]
which means that $\tilde{S}_{m-1}$ is negative definite.
\end{proof}

\begin{lemma} Let $S_m=V_m^TAV_m$ be the compression of $A$. With $E=[ \vec{0}, I_{m-1} ] \in \mathbb{R}^{(m-1)\times m}$, let $\widehat{H}_{m-1}=EH_mE^T$. Then $\widehat{H}_{m-1}$ is invertible and one may compute the compression $S_m$ by the values computed in Algorithm~\ref{alg:arnlike} as follows
\[
 S_m=
 \begin{bmatrix}
     0   &         0         \\
 \vec{u} &  \tilde{S}_{m-1}  \\
 \end{bmatrix}, ~ \vec{u}=\frac{\|\bsym\|}{\|\vec{b}\|}\vec{e}_1 \in \mathbb{R}^{m-1}, ~ \mbox{and} ~\,
        \tilde{S}_{m-1}=\gamma I-\widehat{H}_{m-1}^{-1}+h_{m+1,m}^2(\vec{v}_{m+1}^TA\vec{v}_{m+1}-\gamma)\widehat{H}_{m-1}^{-1}\vec{e}_{m-1}\vec{e}_{m-1}^T\widehat{H}_{m-1}^{-1}\,.
\]
\end{lemma}
\begin{proof}
According to Lemma~\ref{Smstruct}, the structure is as stated in this lemma where
\[
  \tilde{S}_{m-1}=W_{m-1}^T\Asym W_{m-1}.
\]
By formula (5.8) in Grimm~\cite{ratkryphi11}, we obtain that $\widehat{H}_{m-1}$ in the relation on the right-hand side of  \eqref{arn:rel} is invertible as well as
\[
\tilde{S}_{m-1}=W_{m-1}^T\Asym W_{m-1}=\gamma I-\widehat{H}_{m-1}^{-1}+\widehat{h}_{m,m-1}^2(\vec{w}_m^T \Asym \vec{w}_m -\gamma) \widehat{H}_{m-1}^{-1}\vec{e}_{m-1}\vec{e}_{m-1}^T\widehat{H}_{m-1}^{-1}.
\]
        As outlined in the discussion following formula \eqref{arn:rel}, the values $H_m$ computed in Algorithm~\ref{alg:arnlike} relate to the ones computed by Algorithm~\ref{alg:arn}, collected in $\widehat{H}_{m-1}$, as $\hat{H}_{m-1}=EH_mE^T$. Therefore, $\widehat{h}_{m,m-1}=h_{m+1,m}$ and $\vec{v}_{m+1}^TA\vec{v}_{m+1}=\vec{w}_m^T RAR^T\vec{w}_m=\vec{w}_m^T \Asym\vec{w}_m$ conclude the proof of our lemma.
\end{proof}
The following theorem states that the boundary pixels are correctly set in the first Krylov step and not altered afterwards due to the properties of the matrix $R^T$.
An alternative representation of the Krylov approximation is given with the help of the $\varphi_1$-function analogous to Lemma~\ref{exactsolrep}.
\begin{theorem} \label{Kryapproxext}
        The Krylov approximation to the matrix exponential times vector, $e^{tA}\vec{b}$, in the extended Krylov subspace $\mathcal{E}_m((\gamma I-A)^{-1},\vec{b})$ reads
\[
\|\vec{b}\| V_m e^{tS_m}\vec{e}_1 = \vec{b} + R^T\left( \|\bsym\| W_{m-1} t \varphi_1(t \tilde{S}_{m-1})\vec{e}_1 \right)\,.
\]
\end{theorem}
\begin{proof}
With the help of Lemma~\ref{Smstruct}, one obtains
\[
 e^{tS_m} = \sum_{k=0}^\infty \frac{t^k}{k!} S_m^k
 =
 \sum_{k=0}^\infty \frac{t^k}{k!}
 \begin{pmatrix}
     0    & 0 \\
  \vec{u} & \tilde{S}_{m-1}
 \end{pmatrix}^k
 =
 I + \sum_{k=1}^\infty \frac{t^k}{k!}
 \begin{pmatrix}
   0 & 0 \\
  \tilde{S}_{m-1}^{k-1} \vec{u} & \tilde{S}_{m-1}^k
 \end{pmatrix}\,,
\]
and hence
\begin{align*}
 e^{tS_m}\vec{e}_1 &= \vec{e}_1
 +
 \sum_{k=1}^\infty \frac{t^k}{k!}
 \begin{pmatrix}
   0 & 0 \\
   \tilde{S}_{m-1}^{k-1} \vec{u} & \tilde{S}_{m-1}^k
 \end{pmatrix} \vec{e}_1
 =
 \vec{e}_1 +
 \begin{pmatrix}
   0 \\
   \sum_{k=1}^\infty \frac{t^k}{k!} \tilde{S}_{m-1}^{k-1}\vec{u}
 \end{pmatrix}
 =
 \begin{pmatrix}
   1 \\
  t\varphi_1(t\tilde{S}_{m-1})\vec{u}
 \end{pmatrix}
 =
\begin{pmatrix}
   1 \\
   \frac{\|\bsym\|}{\|\vec{b}\|} t\varphi_1(t\tilde{S}_{m-1})\vec{e}_1
 \end{pmatrix}
\end{align*}
Finally, one obtains
\begin{align*}
\|\vec{b}\|V_m e^{tS_m}\vec{e}_1
 &= \|\vec{b}\|
 \begin{bmatrix}
 \vec{v}_1 & R^TW_{m-1}
 \end{bmatrix}
 \begin{bmatrix}
   1 \\
  \frac{\|\bsym\|}{\|\vec{b}\|} t\varphi_1(t\tilde{S}_{m-1})\vec{e}_1
 \end{bmatrix}
 =
\|\vec{b}\|\vec{v}_1 + \|\vec{b}\|R^TW_{m-1} \left( \frac{\|\bsym\|}{\|\vec{b}\|} t\varphi_1(t\tilde{S}_{m-1})\vec{e}_1 \right) \\
 &=
 \vec{b} + R^T \left( \|\bsym\| W_{m-1} t\varphi_1(t\tilde{S}_{m-1})\vec{e}_1\right).
\end{align*}
\end{proof}
The comparison of the exact solution with the Krylov approximation via the alternative representation by the $\varphi_1$-function leads to the error bound in Theorem~\ref{error}. The $\varphi_1$-function can be approximated uniformly for matrices/operators with field-of-values in the left half-plane (cf. \cite{ratkryphi11}). Here we obtain even better bounds due to the symmetry of the matrix at which the function is evaluated.
\begin{theorem}\label{error}
        The error of the rational Krylov approximation $\vec{f}_m=\|\vec{b}\|V_me^{tS_m}\vec{e}_1$ in the extended Krylov subspace $\mathcal{E}_{m-1}((\gamma I-tA)^{-1},\vec{b})$ with $\gamma >0$ to the solution of \eqref{ode} reads, for $m \geq 2$ and $t > 0$,
\[
  \|e^{tA}\vec{b}-\vec{f}_m\| \leq 2tE_m(\gamma)\|\bsym\|,
\]
with
\[
  E_m(\gamma)=\min_{r \in \mathcal{R}_{m-1}} \|\varphi_1-r\|_{(-\infty,0]}\,,
\]
where $\|\cdot \|_{(-\infty,0]}$ designates the supremum norm on $(-\infty,0]$ and the space
\[
 \mathcal{R}_{m-1}=
   \left\{
           \frac{p_{m-2}}{(\gamma- \cdot )^{m-2}} ~|~ p_{m-2} \in \mathcal{P}_{m-2}
   \right\}
\]
is the space of rational functions of the indicated form and dimension $m-1$. Here, $\mathcal{P}_{m-2}$ is the space of polynomials with degree less than or equal to $m-2$.
\end{theorem}
\begin{proof}
With Lemma~\ref{exactsolrep} and Theorem~\ref{Kryapproxext}, one obtains
\[
        e^{tA}\vec{b}-\vec{f}_m
        = tR^T(\varphi_1(t\Asym)\bsym-\|\bsym\|W_{m-1}\varphi_1(t\tilde{S}_{m-1})\vec{e}_1)\,.
\]
Hence, by the exactness property in the Krylov space $\mathcal{K}_{m-1}((\gamma I-t\Asym)^{-1}, \bsym)$, namely
\[
        r(t\Asym)\bsym = \|\bsym\|W_{m-1}r(t\tilde{S}_{m-1})\vec{e}_1
\]
for all $r \in \mathcal{R}_{m-1}$, one obtains
\begin{align*}
        \|e^{tA}\vec{b}-\vec{f}_m \|
        &\leq t\|(\varphi_1(t\Asym)-r(t\Asym))\bsym\|
        + t\|\bsym\|\|(r(t\tilde{S}_{m-1})-\varphi_1(t\tilde{S}_{m-1}))\vec{e}_1\| \\
        &\leq 2t\|\bsym\|\cdot \min_{r \in \mathcal{R}_{m-1}} \max_{z \in (-\infty,0]}
        |\varphi_1(z)-r(z)|\,.
\end{align*}
\end{proof}
Note that $E_m(\gamma)$ does neither depend on the chosen $t>0$ nor on the size of the matrix $A$. With exactly the same ideas as in van den Eshof and Hochbruck~\cite{marlis_jasper}, Table~\ref{tab:errphi} of optimal values of $\gamma$ with respect to the minimisation of the error can be numerically computed with the help of a simple transform and the Remez algorithm.
Note that there is a subtle issue about scaling. One has to use the Krylov subspace as given in the theorem. Alternatively, one might use the space $\mathcal{E}_{m}((\tilde{\gamma}I-A)^{-1},\vec{b})$ with $\tilde{\gamma}=\gamma_{\mbox{\scriptsize opt}} \slash t$ for the simple reason that
$\mathcal{E}_m((\gamma I-tA)^{-1},\vec{b})=\mathcal{E}_{m}((\frac{\gamma}{t}I-A)^{-1},\vec{b})$. We refer the reader to van den Eshof and Hochbruck~\cite{marlis_jasper} for details.
\begin{table}[tbhp]{
  \caption{Numerical approximation to the optimal value of $\gamma$, $\gamma_{\mbox{\scriptsize opt}}$, and the error $E_m(\gamma_{\mbox{\scriptsize opt}})$ for the optimal choice of $\gamma$, dimension $m$ of the extended Krylov subspace, and $\#lss$ linear system solves by the multigrid method.}\label{tab:errphi}.
\begin{center}
\begin{tabular}{c|c|cc||c|c|cc}
\hline
\rule{0pt}{3ex}
 $\#lss$ & $m$ & $E_m(\gamma_{\mbox{\scriptsize opt}})$ & $\gamma_{\mbox{\scriptsize opt}}$ &
 $\#lss$ & $m$ & $E_m(\gamma_{\mbox{\scriptsize opt}})$ & $\gamma_{\mbox{\scriptsize opt}}$ \\[0.5ex]
\hline
 $1$ & $3$  & $2.6\cdot10^{-2}$ &  $1.5$   & $11$ & $13$ &  $5.3\cdot10^{-7}$ &   $8.5$  \rule{0pt}{3ex}\\
 $2$ & $4$  & $6.6\cdot10^{-3}$ &  $3.5$   & $12$ & $14$ &  $1.8\cdot10^{-7}$ &   $10$  \\
 $3$ & $5$  & $2.2\cdot10^{-3}$ &  $5.5$   & $13$ & $15$ &  $5.7\cdot10^{-8}$ &  $11.5$  \\
 $4$ & $6$  & $6.9\cdot10^{-4}$ &  $3.5$   & $14$ & $16$ &  $2.5\cdot10^{-8}$ &   $10$  \\
 $5$ & $7$  & $2.0\cdot10^{-4}$ &   $5$    & $15$ & $17$ &  $8.6\cdot10^{-9}$ &  $11.5$  \\
 $6$ & $8$  & $8.9\cdot10^{-5}$ &   $7$    & $16$ & $18$ &  $3.1\cdot10^{-9}$ &   $13$  \\
 $7$ & $9$  & $2.8\cdot10^{-5}$ &  $8.5$   & $17$ & $19$ &  $1.3\cdot10^{-9}$ &  $11.5$  \\
 $8$ & $10$ & $1.0\cdot10^{-5}$ &  $6.5$   & $18$ & $20$ & $4.8\cdot10^{-10}$ &   $13$  \\
 $9$ & $11$ & $3.8\cdot10^{-6}$ &  $8.5$   & $19$ & $21$ & $1.9\cdot10^{-10}$ &  $14.5$  \\
$10$ & $12$ & $1.1\cdot10^{-6}$ &   $10$   & $20$ & $22$ & $8.3\cdot10^{-11}$ &   $16$  \\
\hline
\end{tabular}
\end{center}
}
\end{table}
We illustrate the bound and the necessity of the scaling numerically. We use an all-white square grey-scale picture of size $1024 \times 1024$. The Canny-like edge detector then produces the mask with all boundary pixels set to one and all interior points set to zero. That is, the compressed picture has a white boundary and all pixels in the interior are black. For this simple example, the solution of the inpainting by the heat equation can be computed at any time $t>0$ by fast transforms. In Figure~\ref{errorbound}, we show the error bound of Theorem~\ref{Kryapproxext} with the optimal choices of $\gamma$ according to Table~\ref{tab:errphi} as a black solid line
and the error of the approximation with respect to the Krylov subspace
$\mathcal{E}_{m}((\tilde{\gamma}I-A)^{-1},\vec{b})$ with $\tilde{\gamma}=\gamma_{\mbox{\scriptsize opt}} \slash t$ as in the theorem as green circle-marked line for $t=25$, $t=10^2$, and $t=10^4$, respectively.
The red diamond-marked line corresponds to the extended Krylov subspace $\mathcal{E}_m((I-A)^{-1},\vec{b})$. That is, $\gamma$ is set to one and not scaled. For $t=25$, the approximation for small dimensions of the space with $\gamma$ fixed to one is clearly worse than the error bound in contrast to the properly scaled Krylov subspace. For $t=10^2$, and $t=10^4$, the approximation of the space with $\gamma$ fixed to one does not improve for larger dimensions of the Krylov subspace, either.
\begin{figure}[!htbp]
\begin{center}
\begin{tikzpicture}[scale=1.0]
\begin{axis}[
width=4.15cm, height=2.65cm, scale only axis,
xmin=2, xmax=23, ymode=log, ymin=2e-10, ymax=2e1,
xtick={3,6,9,12,15,18, 21}, 
]
\addplot[color=black, very thick, solid] coordinates {
(  3  ,  1.3  )
(  4  ,  0.33  )
(  5  ,  0.11  )
(  6  ,  0.034499999999999996  )
(  7  ,  0.01  )
(  8  ,  0.00445  )
(  9  ,  0.0014  )
(  10  ,  0.0005  )
(  11  ,  0.00019  )
(  12  ,  5.5e-05  )
(  13  ,  2.65e-05  )
(  14  ,  9e-06  )
(  15  ,  2.8500000000000002e-06  )
(  16  ,  1.2499999999999999e-06  )
(  17  ,  4.4500000000000003e-07  )
(  18  ,  1.55e-07  )
(  19  ,  6.5e-08  )
(  20  ,  2.4e-08  )
(  21  ,  9.499999999999999e-09  )
(  22  ,  4.15e-09  )
};
\addplot[color=green, mark=*, thick, solid] coordinates {
(  3  ,  0.21691151922347554  )
(  4  ,  0.06504043476433637  )
(  5  ,  0.0202147097615134  )
(  6  ,  0.006253419733285197  )
(  7  ,  0.0015148266124797123  )
(  8  ,  0.000578614747132046  )
(  9  ,  0.00014246171152245324  )
(  10  ,  4.9611830164743985e-05  )
(  11  ,  1.64622620853198e-05  )
(  12  ,  5.935219362359168e-06  )
(  13  ,  2.039954872164645e-06  )
(  14  ,  7.431910517029416e-07  )
(  15  ,  3.289278486834353e-07  )
(  16  ,  1.186510232870206e-07  )
(  17  ,  4.673208929714668e-08  )
(  18  ,  1.9184587700545117e-08  )
(  19  ,  7.982883436305909e-09  )
(  20  ,  2.8960031508774216e-09  )
(  21  ,  1.0471551814985704e-09  )
(  22  ,  3.668536230855963e-10  )
};
\addplot[color=red, mark=diamond*, mark options={scale=1.25}, thick, solid] coordinates {
(  3  ,  1.0753977698385597  )
(  4  ,  0.7164898869472262  )
(  5  ,  0.41777620122695036  )
(  6  ,  0.21720931112024922  )
(  7  ,  0.10255752727356207  )
(  8  ,  0.04399878702586424  )
(  9  ,  0.017066053906920394  )
(  10  ,  0.005796100404134619  )
(  11  ,  0.0015922692490613498  )
(  12  ,  0.00030900682344624323  )
(  13  ,  2.65345101970989e-05  )
(  14  ,  1.0920270415985245e-05  )
(  15  ,  4.4497065022865834e-06  )
(  16  ,  5.984926147930467e-07  )
(  17  ,  1.9999492250096696e-07  )
(  18  ,  7.954297563380616e-08  )
(  19  ,  7.064883803247464e-09  )
(  20  ,  6.147158319594633e-09  )
(  21  ,  1.2162279812252343e-09  )
(  22  ,  5.181078984158503e-10  )
};
\end{axis}
\end{tikzpicture}%
\hspace{0.5cm}
\begin{tikzpicture}[scale=1.0]
\begin{axis}[
width=4.15cm, height=2.65cm, scale only axis,
xmin=2, xmax=23, ymode=log, ymin=2e-10, ymax=7e1,
xtick={3,6,9,12,15,18, 21}, 
]
\addplot[color=black, very thick, solid] coordinates {
(  3  ,  5.2  )
(  4  ,  1.32  )
(  5  ,  0.44  )
(  6  ,  0.13799999999999998  )
(  7  ,  0.04  )
(  8  ,  0.0178  )
(  9  ,  0.0056  )
(  10  ,  0.002  )
(  11  ,  0.00076  )
(  12  ,  0.00022  )
(  13  ,  0.000106  )
(  14  ,  3.6e-05  )
(  15  ,  1.1400000000000001e-05  )
(  16  ,  4.9999999999999996e-06  )
(  17  ,  1.7800000000000001e-06  )
(  18  ,  6.2e-07  )
(  19  ,  2.6e-07  )
(  20  ,  9.6e-08  )
(  21  ,  3.7999999999999996e-08  )
(  22  ,  1.66e-08  )
};
\addplot[color=green, mark=*, thick, solid] coordinates {
(  3  ,  0.4849676250876442  )
(  4  ,  0.15581150241042688  )
(  5  ,  0.05954542777562138  )
(  6  ,  0.01691642390587037  )
(  7  ,  0.0051718239154269555  )
(  8  ,  0.0023381921817699188  )
(  9  ,  0.0006760195358878612  )
(  10  ,  0.00020334471900665903  )
(  11  ,  9.988299941280762e-05  )
(  12  ,  3.076160423936319e-05  )
(  13  ,  1.346661445874937e-05  )
(  14  ,  4.4187808393856715e-06  )
(  15  ,  1.3804735201105715e-06  )
(  16  ,  5.97757717843216e-07  )
(  17  ,  1.9607592180373394e-07  )
(  18  ,  6.297638126248771e-08  )
(  19  ,  2.6602594734169195e-08  )
(  20  ,  9.086411529145418e-09  )
(  21  ,  3.1326299023979616e-09  )
(  22  ,  1.1372753457629054e-09  )
};
\addplot[color=red, mark=diamond*, mark options={scale=1.25}, thick, solid] coordinates {
(  3  ,  2.03123724446581  )
(  4  ,  1.7543430285260235  )
(  5  ,  1.4865932135682318  )
(  6  ,  1.2293426985512967  )
(  7  ,  0.9876894891316128  )
(  8  ,  0.7719481482512358  )
(  9  ,  0.5906568854653788  )
(  10  ,  0.44275223625175764  )
(  11  ,  0.32327530205698696  )
(  12  ,  0.22984579135666805  )
(  13  ,  0.15995695815316868  )
(  14  ,  0.10959282368809087  )
(  15  ,  0.07402727953846189  )
(  16  ,  0.048950899812824286  )
(  17  ,  0.03146272177221106  )
(  18  ,  0.0196838656003996  )
(  19  ,  0.012060178559753235  )
(  20  ,  0.007264696657238098  )
(  21  ,  0.00428302315258767  )
(  22  ,  0.00244995120605202  )
};
\end{axis}
\end{tikzpicture}%
\hspace{0.5cm}
\begin{tikzpicture}[scale=1.0]
\begin{axis}[
width=4.15cm, height=2.65cm, scale only axis,
xmin=2, xmax=23, ymode=log, ymin=2e-9, ymax=5e2,
xtick={3,6,9,12,15,18, 21}, 
]
\addplot[color=black, very thick, solid] coordinates {
(  3  ,  520.0  )
(  4  ,  132.0  )
(  5  ,  44.0  )
(  6  ,  13.799999999999999  )
(  7  ,  4.0  )
(  8  ,  1.7799999999999998  )
(  9  ,  0.5599999999999999  )
(  10  ,  0.2  )
(  11  ,  0.076  )
(  12  ,  0.022000000000000002  )
(  13  ,  0.0106  )
(  14  ,  0.0036  )
(  15  ,  0.00114  )
(  16  ,  0.0005  )
(  17  ,  0.00017800000000000002  )
(  18  ,  6.2e-05  )
(  19  ,  2.6000000000000002e-05  )
(  20  ,  9.6e-06  )
(  21  ,  3.7999999999999996e-06  )
(  22  ,  1.66e-06  )
};
\addplot[color=green, mark=*, thick, solid] coordinates {
(  3  ,  2.0753084434133093  )
(  4  ,  0.7443413357272933  )
(  5  ,  0.31502477782331145  )
(  6  ,  0.08290537402865651  )
(  7  ,  0.028928153739370594  )
(  8  ,  0.013718338029902494  )
(  9  ,  0.004829425690469551  )
(  10  ,  0.00134034267650733  )
(  11  ,  0.0006075913384862188  )
(  12  ,  0.0002198054501934781  )
(  13  ,  6.134563861233584e-05  )
(  14  ,  2.7280858546024892e-05  )
(  15  ,  1.1469567721598778e-05  )
(  16  ,  2.551150528958837e-06  )
(  17  ,  1.1612437140787086e-06  )
(  18  ,  5.375894185835173e-07  )
(  19  ,  1.1620852940441444e-07  )
(  20  ,  5.3089083197304324e-08  )
(  21  ,  2.461397267314388e-08  )
(  22  ,  9.82248953590402e-09  )
};
\addplot[color=red, mark=diamond*, mark options={scale=1.25}, thick, solid] coordinates {
(  3  ,  7.952783396345569  )
(  4  ,  7.861093799126167  )
(  5  ,  7.770090964745885  )
(  6  ,  7.679623060470789  )
(  7  ,  7.590124386530171  )
(  8  ,  7.5024999874814  )
(  9  ,  7.417367305746275  )
(  10  ,  7.333195202007067  )
(  11  ,  7.247579483945251  )
(  12  ,  7.160505841877801  )
(  13  ,  7.073483404109842  )
(  14  ,  6.987897685857794  )
(  15  ,  6.903861636648133  )
(  16  ,  6.81976679753056  )
(  17  ,  6.734483742173056  )
(  18  ,  6.6486859000136  )
(  19  ,  6.5636599157243705  )
(  20  ,  6.479980230280394  )
(  21  ,  6.396753167391087  )
(  22  ,  6.312788077517177  )
};
\end{axis}
\end{tikzpicture}%
\end{center}
        \caption{Error bound of Theorem~\ref{Kryapproxext} (black) solid line for the optimal choice of $\gamma$, error of the extended Krylov subspace approximation to the compressed all-white picture for the optimal $\gamma$ (green) circle-marked line, and error of the extended Krylov subspace with $\gamma$ fixed to one (red) diamond-marked line for $t=25$ on left-hand side, $t=10^2$ in the middle, $t=10^4$ on the right-hand side.} \label{errorbound}
\end{figure}
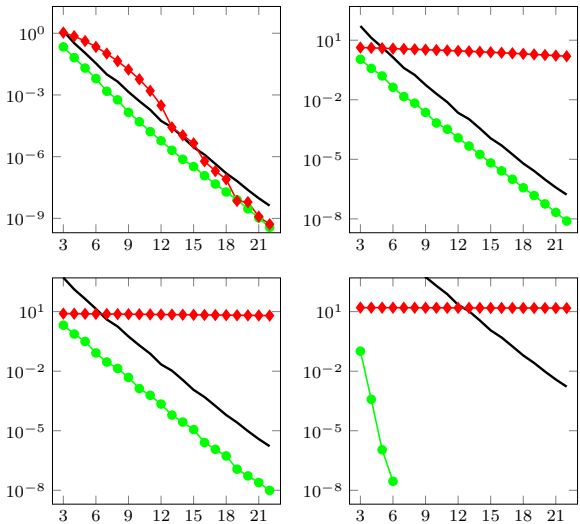
Since the eye can not distinguish very small deviations, the error bound in Theorem~\ref{Kryapproxext} and our experiment suggest that an extended Krylov subspace of a small dimension and with the optimal choice of $\gamma$ is sufficient for image processing purposes. This will be confirmed in Section~\ref{sec:numex} by numerical experiments with real pictures.


\section{Implementation details} \label{sec:multigrid}
For the efficient computation of $\vec{v}^{\vec{h}}=(\gamma I-A^{\vec{h}})^{-1}\vec{b}^{\vec{h}}$, we use the multigrid method applied to the system
\[
        B^{\vec{h}}\vec{v}^{\vec{h}}=\vec{b}^{\vec{h}}, \qquad  B^{\vec{h}} = (\gamma I-A^{\vec{h}})\,.
\]
The superscript $\vec{h}$ indicates that the matrices and vectors belong to the finest grid which corresponds to the original image.
In order to efficiently implement the multigrid method, we operate on discrete images (cf.~\cite{Mainbergeretal11,Bruhnetal05}).
For the application of the multigrid method, we look at a fine grid $\Omega^\vec{h}$ with $N^\vec{h}=N_x^\vec{h} \times N_y^\vec{h}$ pixels, where $N_x^\vec{h}$ and $N_y^\vec{h}$ correspond to the number of pixels in $x$- and $y$-direction, respectively. The grid spacing is denoted by $\vec{h}=(h_x,h_y)^T$.
On the finest grid $(h_x,h_h)^T=(1,1)^T$, which is a popular choice in image processing. For the next coarser grid, one would like to double the grid spacing in both directions. This is only possible for powers of two. In order to include other grids, we define the coarser
grid $\Omega^\vec{H}$ with the spacings $\vec{H}=(H_x,H_y)^T$, with
\[
  H_x=h_x\frac{N_x^\vec{h}}{N_x^\vec{H}}, \qquad H_y=h_y\frac{N_y^\vec{h}}{N_y^\vec{H}}\,,
\]
where $N_x^\vec{H}= \lceil N_x^\vec{h} \slash 2 \rceil$ and $N_y^\vec{H}=\lceil N_y^\vec{h} \slash 2\rceil$ are the number of pixels in each direction in the coarse grid. Thereby, $\lceil \cdot \rceil$ denotes the ceiling function.
For the restriction, the coarse pixel is the average of the fine pixels according to the area of the fine pixel that contributes to the coarse pixel. The prolongation reverses this process. For the restriction matrix $I_\vec{h}^\vec{H}$ and the prolongation matrix $I_\vec{H}^\vec{h}$, one has the relation
\[
  I_\vec{h}^\vec{H}=\alpha(I_\vec{H}^\vec{h})^T, \qquad
  \alpha=\frac{N_x^\vec{H}N_y^\vec{H}}{N_x^\vec{h}N_y^\vec{h}}=\frac{h_xh_y}{H_xH_y}\,.
\]
Restriction and prolongation are illustrated in Figure~\ref{prolres}.
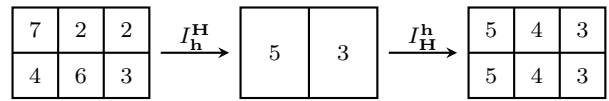
\begin{figure}[!htbp]
\begin{center}
\begin{tikzpicture}[scale=1.0]
  \draw[thick] (1,1) rectangle (4,3);
  \draw[thick] (1,2) -- (4,2);
  \draw[thick] (2,1) -- (2,3);
  \draw[thick] (3,1) -- (3,3);
  \node at (1.5,1.5) {$4$};
  \node at (2.5,1.5) {$6$};
  \node at (3.5,1.5) {$3$};
  \node at (1.5,2.5) {$7$};
  \node at (2.5,2.5) {$2$};
  \node at (3.5,2.5) {$2$};
  \draw[thick,->,>=stealth] (4.25,2) -- (5.75,2);
  \node at (5,2.35) {$I_\vec{h}^{\bf H}$};
  \draw[thick] (6,1) rectangle (9,3);
  \draw[thick] (7.5,1) -- (7.5,3);
  \node at (6.75,2) {$5$};
  \node at (8.25,2) {$3$};
 \draw[thick,->,>=stealth] (9.25,2) -- (10.75,2);
  \node at (10,2.35) {$I_{\bf H}^{\bf h}$};
  \draw[thick] (11,1) rectangle (14,3);
  \draw[thick] (11,2) -- (14,2);
  \draw[thick] (12,1) -- (12,3);
  \draw[thick] (13,1) -- (13,3);
  \node at (11.5,1.5) {$5$};
  \node at (12.5,1.5) {$4$};
  \node at (13.5,1.5) {$3$};
  \node at (11.5,2.5) {$5$};
  \node at (12.5,2.5) {$4$};
  \node at (13.5,2.5) {$3$};
\end{tikzpicture}
\caption{Example for the restriction and prolongation} \label{prolres}
\end{center}
\end{figure}
Since we have two different sorts of pixels, we also have to apply the restriction and prolongation to our binary inpainting mask $\vec{c}$. We therefore adapt the restriction of the inpainting by applying the element-wise sign function to the restricted inpainting mask
\[
 \vec{c}^\vec{H}=\sgn (I_\vec{h}^\vec{H}\vec{c}^\vec{h}).
\]
This is also illustrated in Figure~\ref{resipmask}.
\begin{figure}[!htbp]
\begin{center}
\begin{tikzpicture}[scale=1.0]
  \draw[thick] (1,1) rectangle (3,3);
  \draw[thick] (1,2) -- (3,2);
  \draw[thick] (2,1) -- (2,3);
  \node at (1.5,1.5) {$0$};
  \node at (2.5,1.5) {$0$};
  \node at (1.5,2.5) {$1$};
  \node at (2.5,2.5) {$0$};
  \node[below] at (2,1) {${\bf c}^{\bf h}$};
  \draw[thick,->,>=stealth] (3.25,2) -- (4.75,2);
  \node at (4,2.35) {$I_{\bf h}^{\bf H}$};
  \draw[thick] (5,1) rectangle (7,3);
  \node at (6,2) {$0.25$};
  \node[below] at (6,1) {$I_{\bf h}^{\bf H}{\bf c}^{\bf h}$};
  \draw[thick,->,>=stealth] (7.25,2) -- (8.75,2);
  \node at (8,2.35) {$\sgn$};
  \draw[thick] (9,1) rectangle (11,3);
  \node at (10,2) {$1$};
  \node[below] at (10,1) {${\bf c}^{\bf H}$};
\end{tikzpicture}
\caption{Example for the restriction applied to the inpainting mask} \label{resipmask}
\end{center}
\end{figure}
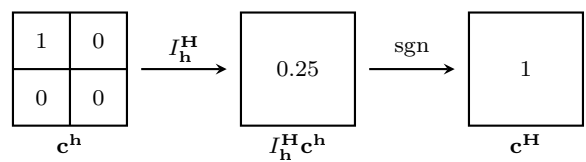
With these two definitions, we obtain a natural definition of the coarse matrix $A^\vec{H}$. We can just use again the standard stencil for the
Laplace operator with respect to the grid spacing $\vec{H}=(H_x,H_y)^T$. For the multigrid method, we also need to compute the restriction of the fine residual $\vec{r}^\vec{h}$. With the Hadamard product $\circ$, the restriction of the fine residual to the coarse residual is
\[
 \vec{r}^\vec{H}=(\overrightarrow{\vec{1}}-\vec{c}^\vec{H}) \circ (I_\vec{h}^\vec{H}\vec{r}^\vec{h}),
   \qquad
 \overrightarrow{\vec{1}}=(1,\cdots,1)^T.
\]
The residual needs to be set to zero for known pixels (known according to the coarse inpainting mask $\vec{c}^\vec{H}$). This is illustrated in Figure~\ref{rescoarsecomp}.
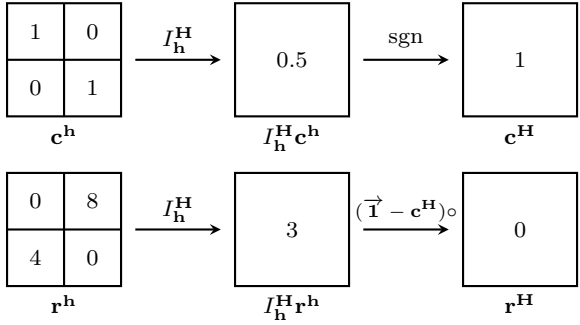
\begin{figure}[!htbp]
\begin{center}
\begin{tikzpicture}[scale=1.0]
  \draw[thick] (1,4) rectangle (3,6);
  \draw[thick] (1,5) -- (3,5);
  \draw[thick] (2,4) -- (2,6);
  \node at (1.5,4.5) {$0$};
  \node at (2.5,4.5) {$1$};
  \node at (1.5,5.5) {$1$};
  \node at (2.5,5.5) {$0$};
  \node[below] at (2,4) {${\bf c}^{\bf h}$};
  \draw[thick,->,>=stealth] (3.25,5) -- (4.75,5);
  \node at (4,5.35) {$I_{\bf h}^{\bf H}$};
  \draw[thick] (5,4) rectangle (7,6);
  \node at (6,5) {$0.5$};
  \node[below] at (6,4) {$I_{\bf h}^{\bf H}{\bf c}^{\bf h}$};
  \draw[thick,->,>=stealth] (7.25,5) -- (8.75,5);
  \node at (8,5.35) {$ \sgn $};
  \draw[thick] (9,4) rectangle (11,6);
  \node at (10,5) {$1$};
  \node[below] at (10,4) {${\bf c}^{\bf H}$};
  \draw[thick] (1,1) rectangle (3,3);
  \draw[thick] (1,2) -- (3,2);
  \draw[thick] (2,1) -- (2,3);
  \node at (1.5,1.5) {$4$};
  \node at (2.5,1.5) {$0$};
  \node at (1.5,2.5) {$0$};
  \node at (2.5,2.5) {$8$};
  \node[below] at (2,1) {${\bf r}^{\bf h}$};
  \draw[thick,->,>=stealth] (3.25,2) -- (4.75,2);
  \node at (4,2.35) {$I_{\bf h}^{\bf H}$};
  \draw[thick] (5,1) rectangle (7,3);
  \node at (6,2) {$3$};
  \node[below] at (6,1) {$I_{\bf h}^{\bf H}{\bf r}^{\bf h}$};
  \draw[thick,->,>=stealth] (7.25,2) -- (8.75,2);
  \node at (8,2.35) {\scriptsize ${\bf (\overrightarrow{{\bf 1}}-c^H)}$$\circ$};
  \draw[thick] (9,1) rectangle (11,3);
  \node at (10,2) {$0$};
  \node[below] at (10,1) {${\bf r}^{\bf H}$};
\end{tikzpicture}
\caption{Computation of the coarse residual ${\bf r}^{\bf H}$. As for the fine residual,
         the coarse residual is set to zero for known pixels.} \label{rescoarsecomp}
\end{center}
\end{figure}
In order to apply nested iteration to obtain a good starting value for the multigrid cycles, we also need a restriction for the right-hand side $\vec{b}^\vec{h}$. Here we use
\[
\vec{b}^\vec{H}=\left(I_\vec{h}^\vec{H}(\vec{c}^\vec{h} \circ \vec{b}^\vec{h}) \right)  {\bf \oslash} \left( I_\vec{h}^\vec{H}\vec{c}^\vec{h}\right)\,,
\]
where $\oslash$ is element-wise division with the exception that a division by zero leads to zero. This is illustrated in Figure~\ref{bcoarsecomp}.
\begin{figure}[!htbp]
\begin{center}
\begin{tikzpicture}[scale=1.0]
  \draw[thick] (1,4) rectangle (3,6);
  \draw[thick] (1,5) -- (3,5);
  \draw[thick] (2,4) -- (2,6);
  \node at (1.5,4.5) {$0$};
  \node at (2.5,4.5) {$1$};
  \node at (1.5,5.5) {$1$};
  \node at (2.5,5.5) {$0$};
  \node[below] at (2,4) {${\bf c}^{\bf h}$};
  \draw[thick,->,>=stealth] (3.25,5) -- (8.75,5);
  \node at (6,5.35) {$I_{\bf h}^{\bf H}$};
  \draw[thick] (9,4) rectangle (11,6);
  \node at (10,5) {$0.5$};
  \node[below] at (10,4) {$I_{\bf h}^{\bf H}{\bf c}^{\bf h}$};
  \draw[thick,->,>=stealth] (11.25,5) -- (12.75,5);
  \node at (12,5.35) {$ \sgn $};
  \draw[thick] (13,4) rectangle (15,6);
  \node at (14,5) {$1$};
  \node[below] at (14,4) {${\bf c}^{\bf H}$};
\draw[thick] (1,1) rectangle (3,3);
  \draw[thick] (1,2) -- (3,2);
  \draw[thick] (2,1) -- (2,3);
  \node at (1.5,1.5) {$2$};
  \node at (2.5,1.5) {$4$};
  \node at (1.5,2.5) {$6$};
  \node at (2.5,2.5) {$0$};
  \node[below] at (2,1) {${\bf b}^{\bf h}$};
  \draw[thick,->,>=stealth] (3.25,2) -- (4.75,2);
  \node at (4,2.35) {${\bf c}^{\bf h}\circ$};
  \draw[thick] (5,1) rectangle (7,3);
  \draw[thick] (5,2) -- (7,2);
  \draw[thick] (6,1) -- (6,3);
  \node at (5.5,1.5) {$0$};
  \node at (6.5,1.5) {$4$};
  \node at (5.5,2.5) {$6$};
  \node at (6.5,2.5) {$0$};
  \node[below] at (6,1) {${\bf c}^{\bf h}\circ{\bf b}^{\bf h}$};
  \draw[thick,->,>=stealth] (7.25,2) -- (8.75,2);
  \node at (8,2.35) {$I_{\bf h}^{\bf H}$};
  \draw[thick] (9,1) rectangle (11,3);
  \node at (10,2) {$2.5$};
  \node[below] at (10,1) {$I_{\bf h}^{\bf H}\left({\bf c}^{\bf h}\circ{\bf b}^{\bf h}\right)$};
  \draw[thick,->,>=stealth] (11.25,2) -- (12.75,2);
  \node at (12,2.35) {\scriptsize ${\bf \oslash} I_{\bf h}^{\bf H}{\bf c}^{\bf h}$};
  \draw[thick] (13,1) rectangle (15,3);
  \node at (14,2) {$5$};
  \node[below] at (14,1) {${\bf b}^{\bf H}$};
\end{tikzpicture}
\caption{Computation of the coarse right-hand side ${\bf b}^{\bf H}$} \label{bcoarsecomp}
\end{center}
\end{figure}
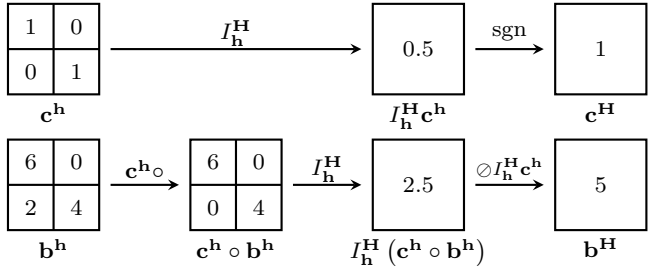
In order to motivate this choice of the restriction for the right-hand side, we illustrate the restriction
of $\vec{b}$ with the reweighting and without in Figure~\ref{Rhs}. Since division
by very small numbers might be instable, the values of $I_\vec{h}^\vec{H}\vec{c}^\vec{h}$ are set to zero
below a certain tolerance before the operator $\oslash$ is applied. More exactly, we use
\[
 c_{ij}^\vec{H}=
 \left\{
   \begin{array}{cl}
     1 & \mbox{if}~(I_\vec{h}^\vec{H}\vec{c}^\vec{h})_{ij}>\epsilon, \\
     0 & \mbox{else}.
   \end{array}
 \right.
\]
\begin{figure}[t]
        \centering
        \subfigure[$\vec{b}^{\vec{h}}, I_{\vec{h}}^{\vec{H}}\vec{b}^{\vec{h}}, \dots$ with normalisation]{\includegraphics[height=3cm]{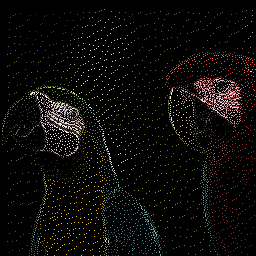} \hspace{0.5cm} \includegraphics[height=3cm]{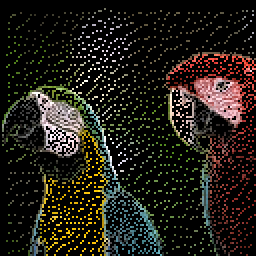} \hspace{0.5cm} \includegraphics[height=3cm]{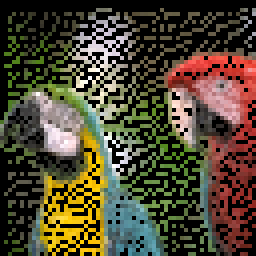} \hspace{0.5cm} \includegraphics[height=3cm]{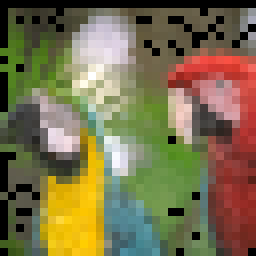}}\\
        \subfigure[$\vec{b}^{\vec{h}}, I_{\vec{h}}^{\vec{H}}\vec{b}^{\vec{h}}, \dots$ without normalisation]{\includegraphics[height=3cm]{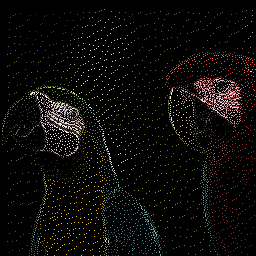} \hspace{0.5cm} \includegraphics[height=3cm]{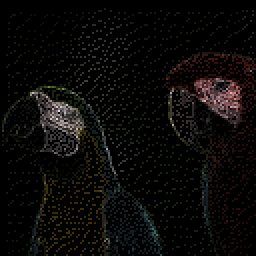} \hspace{0.5cm} \includegraphics[height=3cm]{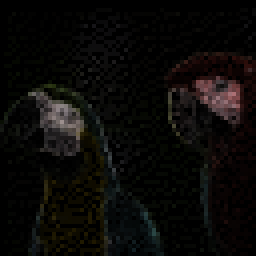} \hspace{0.5cm} \includegraphics[height=3cm]{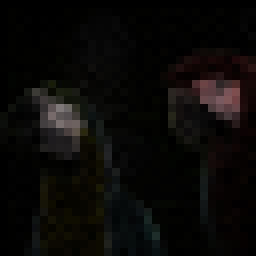}}
        \caption{Restriction of the right-hand side $\vec{b}^{\vec{h}}$. The top row shows the restriction with the reweighting, the bottom row without. Without reweighting, the colours loose intensity on the coarser grids.}
        \label{Rhs}
\end{figure}
With these preparations, we state the full multigrid method as Algorithm~\ref{alg:fmg}. It consists of nested iteration (cf. Algorithm~\ref{alg:nestiter}) for a good starting vector followed by several $\mu$-cycles (cf. Algorithm~\ref{alg:mucycle}). The $\mu$-cycle for $\mu=1$ is also called $V$-cycle and the $\mu$-cycle with $\mu=2$ is also called $W$-cycle.
\begin{algorithm}[ht]
\caption{$\mu$-cycle ($\vec{v}^\vec{h}={\tt MG}(\vec{v}^\vec{h},\vec{b}^\vec{h})$)}
\label{alg:mucycle}
\begin{algorithmic}
\If {$\Omega^\vec{h}=$ coarsest grid}
\State{solve $B^\vec{h}\vec{v}^\vec{h}=\vec{b}^\vec{h}$ by direct solver}
\Else
\State{relax $\nu_1$ times $B^\vec{h}\vec{u}^\vec{h}=\vec{b}^\vec{h}$ with given starting value $\vec{v}^\vec{h}$}
\State{$\vec{r}^\vec{H}=(\overrightarrow{{\bf 1}}-\vec{c}^\vec{H})
                  \circ
                  I_\vec{h}^\vec{H}(\vec{b}^\vec{h}-B^\vec{h}\vec{v}^\vec{h})$}
\State{$\vec{v}^\vec{H}=0$}
\State{$\vec{v}^\vec{H}={\tt MG}(\vec{v}^\vec{H},\vec{r}^\vec{H})$ $\mu$ times}
\State{correct $\vec{v}^\vec{h}=\vec{v}^\vec{h}+I_\vec{H}^\vec{h}\vec{v}^\vec{H}$}
\State{relax $\nu_2$ times $B^\vec{h}\vec{u}^\vec{h}=\vec{b}^\vec{h}$ with starting value $\vec{v}^\vec{h}$}
\EndIf
\end{algorithmic}
\end{algorithm}
\begin{algorithm}[ht]
\caption{Nested iteration ($\vec{v}^\vec{h}={\tt NI}(\vec{b}^\vec{h})$)}
\label{alg:nestiter}
\begin{algorithmic}
\If {$\Omega^\vec{h}=$ coarsest grid}
\State {solve $B^\vec{h}\vec{v}^\vec{h}=\vec{b}^\vec{h}$ by direct solver}
\Else
\State{$\vec{b}^\vec{H}=(I_\vec{h}^\vec{H}(\vec{c}^\vec{h}\circ\vec{b}^\vec{h})\oslash
                 (I_\vec{h}^\vec{H}\vec{c}^\vec{h})$}
\EndIf
\State{$\vec{v}^\vec{H}={\tt NI}(\vec{b}^\vec{H})$}
\State{prolongate $\vec{v}^\vec{h}=I_\vec{H}^\vec{h}\vec{v}^\vec{H}$\\
        $\vec{v}^\vec{h}={\tt MG}(\vec{v}^\vec{h},\vec{b}^\vec{h})$ $\nu_0$ times}
\end{algorithmic}
\end{algorithm}
\begin{algorithm}[ht]
\caption{Full multigrid ($\vec{v}^\vec{h}={\tt FMG}(\vec{v}^\vec{h},\vec{b}^\vec{h})$)}
\label{alg:fmg}
\begin{algorithmic}
\State{
  $\vec{v}^\vec{h}={\tt NI}(\vec{b}^\vec{h})$\\
  $\vec{v}^\vec{h}={\tt MG}(\vec{v}^\vec{h},\vec{b}^\vec{h})$ $k$ times
}
\end{algorithmic}
\end{algorithm}


\section{Numerical experiments}\label{sec:numex}

In this section, we present some experiments with our decoding scheme. In the first subsection, we show that the new method outperforms standard time-integration methods. That the method compares to other (linear) edge-compressing schemes is shown in the second subsection. Finally, we demonstrate the use of our scheme on a real-world device.

\subsection{Performance of the integrator}

Basically, we have to compute the solution $\vec{y}(t)$ of the system of ordinary differential equations \eqref{ode} for a large time $t$. After subdividing the interval $[0,t]$ in $n$ subintervals, the standard and most-used methods to approximate this solution are the implicit (or backward) Euler method
\begin{equation} \label{ieulmeth}
        \vec{y}(t) \approx  \left( \gamma (\gamma I -A)^{-1} \right)^n\vec{b}_0, \quad \gamma = \frac{n}{t}\,
\end{equation}
and the Crank--Nicolson method
\begin{equation} \label{CNmeth}
        \vec{y}(t) \approx \left( (\gamma I+A)(\gamma I-A)^{-1} \right)^n\vec{b}, \quad \gamma = \frac{2n}{t}\,.
\end{equation}
The larger $n$, the more accurate is the approximation.
To apply both methods, we have so solve $n$ linear systems of exactly the same type as for the Krylov method. Since this is the largest workload, we compare the methods with respect to the number of necessary solutions of linear systems of this type. For our edge-compressed all-white square test picture of Section~\ref{sec:decoding}, the relative error in the Euclidean norm is shown in Figure~\ref{errorcomp}. For $t=25$, $t=10^2$, and $t=10^4$, the error of the methods versus the number of necessary solutions of the large linear system are shown. For larger $t$, the Crank-Nicolson method becomes worse, (which is a known behaviour due to stability considerations), while the backward Euler scheme remains unaffected. For large $t$ and an approximation error of about $10^{-3}$, the implicit Euler scheme needs to solve $1000$ linear systems of the type $(\gamma I-A)\vec{x}=\vec{b}$, while the Krylov method only needs $8$. This is a factor of $125$ times faster.
This clearly demonstrates our main contribution that the Krylov method can solve homogeneous inpainting problems with a significantly improved speed.

\begin{figure}[!htbp]
\begin{center}
\begin{center}
\begin{tikzpicture}[scale=1.0]
\begin{axis}[
width=4.15cm, height=2.65cm, scale only axis,
xmin=0.5, xmax=2000, xmode=log, ymode=log, ymin=1e-10, ymax=9,
xtick={1,10,100,1000}, 
]
\addplot[color=green, very thick, solid] coordinates {
(  1  ,  0.2170175106880071  )
(  2  ,  0.06507221606835387  )
(  3  ,  0.020224587460507894  )
(  4  ,  0.006256475399091028  )
(  5  ,  0.0015155668160923065  )
(  6  ,  0.0005788974809516589  )
(  7  ,  0.00014253132392183572  )
(  8  ,  4.963607243525563e-05  )
(  9  ,  1.6470306179162817e-05  )
(  10  ,  5.938119533232197e-06  )
(  11  ,  2.040951687871985e-06  )
(  12  ,  7.435541850991014e-07  )
(  13  ,  3.29088656565675e-07  )
(  14  ,  1.1870906365665412e-07  )
(  15  ,  4.6754791744648283e-08  )
(  16  ,  1.919367127594331e-08  )
(  17  ,  7.987006815719568e-09  )
(  18  ,  2.898049252443851e-09  )
(  19  ,  1.048919488997872e-09  )
(  20  ,  3.637930660976888e-10  )
};
\addplot[color=red, mark=diamond*, mark options={scale=1.5}, thick, solid] coordinates {
(  1.0  ,  0.2822546676303295  )
(  2.0  ,  0.15471137746586308  )
(  4.0  ,  0.08114778799469276  )
(  10.0  ,  0.03340480382460082  )
(  21.0  ,  0.016065026213033913  )
(  46.0  ,  0.007369548373779786  )
(  100.0  ,  0.0033973813227506536  )
(  215.0  ,  0.0015817407646512534  )
(  464.0  ,  0.0007332564444985807  )
(  1000.0  ,  0.0003403034961808725  )
};
\addplot[color=orange, mark=square*, thick, solid] coordinates {
(  1.0  ,  0.7534118338631428  )
(  2.0  ,  0.3831238518607594  )
(  4.0  ,  0.11910356259440875  )
(  10.0  ,  0.0016665017853574836  )
(  21.0  ,  0.0002545678091066715  )
(  46.0  ,  5.305998004641422e-05  )
(  100.0  ,  1.1227709305544703e-05  )
(  215.0  ,  2.4289356420932207e-06  )
(  464.0  ,  5.215036471519117e-07  )
(  1000.0  ,  1.1227769128608907e-07  )
};
\end{axis}
\end{tikzpicture}%
\hspace{0.5cm}
\begin{tikzpicture}[scale=1.0]
\begin{axis}[
width=4.15cm, height=2.65cm, scale only axis,
xmin=0.5, xmax=2000, xmode=log, ymode=log, ymin=1e-10, ymax=9,
xtick={1,10,100,1000}, 
]
\addplot[color=green, very thick, solid] coordinates {
(  1  ,  0.4852045992650753  )
(  2  ,  0.1558876380134189  )
(  3  ,  0.05957452400386078  )
(  4  ,  0.01692468993329022  )
(  5  ,  0.005174351071273447  )
(  6  ,  0.002339334714206753  )
(  7  ,  0.000676349865604398  )
(  8  ,  0.00020344408113477308  )
(  9  ,  9.993180603257921e-05  )
(  10  ,  3.077663548288443e-05  )
(  11  ,  1.3473194899607839e-05  )
(  12  ,  4.420940122910095e-06  )
(  13  ,  1.3811481263220174e-06  )
(  14  ,  5.980497289466297e-07  )
(  15  ,  1.9617168360009294e-07  )
(  16  ,  6.300715317898217e-08  )
(  17  ,  2.6615598991409613e-08  )
(  18  ,  9.090832591003947e-09  )
(  19  ,  3.134098650329857e-09  )
(  20  ,  1.1377422588513751e-09  )
};
\addplot[color=red, mark=diamond*, mark options={scale=1.5}, thick, solid] coordinates {
(  1.0  ,  0.40126544346476506  )
(  2.0  ,  0.21970413005922093  )
(  4.0  ,  0.11518764149419394  )
(  10.0  ,  0.04740612194917566  )
(  21.0  ,  0.02279667486299723  )
(  46.0  ,  0.01045715405364893  )
(  100.0  ,  0.004820688045048943  )
(  215.0  ,  0.002244380276217178  )
(  464.0  ,  0.0010404359410252657  )
(  1000.0  ,  0.000482864293464455  )
};
\addplot[color=orange, mark=square*, thick, solid] coordinates {
(  1.0  ,  1.2316782347949817  )
(  2.0  ,  0.7540768639297181  )
(  4.0  ,  0.38639058371143814  )
(  10.0  ,  0.06493938165994598  )
(  21.0  ,  0.0008945956479579306  )
(  46.0  ,  7.543337634632812e-05  )
(  100.0  ,  1.596202923033302e-05  )
(  215.0  ,  3.4531309570059538e-06  )
(  464.0  ,  7.414031161860299e-07  )
(  1000.0  ,  1.5962117743051398e-07  )
};
\end{axis}
\end{tikzpicture}%
\hspace{0.5cm}
\begin{tikzpicture}[scale=1.0]
\begin{axis}[
width=4.15cm, height=2.65cm, scale only axis,
xmin=0.5, xmax=2000, xmode=log, ymode=log, ymin=1e-10, ymax=90,
xtick={1,10,100,1000}, 
]
\addplot[color=green, very thick, solid] coordinates {
(  1  ,  2.076322520406821  )
(  2  ,  0.7447050500591442  )
(  3  ,  0.3151787112690075  )
(  4  ,  0.08294588484173962  )
(  5  ,  0.02894228916837873  )
(  6  ,  0.013725041347199492  )
(  7  ,  0.0048317855370263955  )
(  8  ,  0.0013409976204391497  )
(  9  ,  0.0006078882314130371  )
(  10  ,  0.00021991285574603445  )
(  11  ,  6.137561449421811e-05  )
(  12  ,  2.7294189042078415e-05  )
(  13  ,  1.1475172201640726e-05  )
(  14  ,  2.552397121033674e-06  )
(  15  ,  1.1618111432423322e-06  )
(  16  ,  5.378521058303557e-07  )
(  17  ,  1.1626531344367031e-07  )
(  18  ,  5.311502460280859e-08  )
(  19  ,  2.4626000024300272e-08  )
(  20  ,  9.827289189029792e-09  )
};
\addplot[color=red, mark=diamond*, mark options={scale=1.5}, thick, solid] coordinates {
(  1.0  ,  1.2720702343155068  )
(  2.0  ,  0.696504208475263  )
(  4.0  ,  0.3651254121707742  )
(  10.0  ,  0.15023390570156445  )
(  21.0  ,  0.07223432192737929  )
(  46.0  ,  0.033132048595160085  )
(  100.0  ,  0.015273027697535625  )
(  215.0  ,  0.007110557610009544  )
(  464.0  ,  0.0032962362110398546  )
(  1000.0  ,  0.0015297698336640309  )
};
\addplot[color=orange, mark=square*, thick, solid] coordinates {
(  1.0  ,  4.390278043215469  )
(  2.0  ,  3.0667132419944556  )
(  4.0  ,  2.119402476483716  )
(  10.0  ,  1.238470551348971  )
(  21.0  ,  0.7296731764407339  )
(  46.0  ,  0.3244359490602846  )
(  100.0  ,  0.06521241981724961  )
(  215.0  ,  0.0006582644329689223  )
(  464.0  ,  2.3492379282799396e-06  )
(  1000.0  ,  5.057805298543146e-07  )
};
\end{axis}
\end{tikzpicture}%
\end{center}
\end{center}
\caption{Error of the implicit Euler scheme (red) diamond-marked line, error of the Crank-Nicolson method (orange) square-marked line, and the rational Krylov subspace method with optimal gamma (green) solid line versus number of linear system solves for $t=25$ on the left-hand side, $t=10^2$ in the middle, $t=10^4$ on the right-hand side.} \label{errorcomp}
\end{figure}
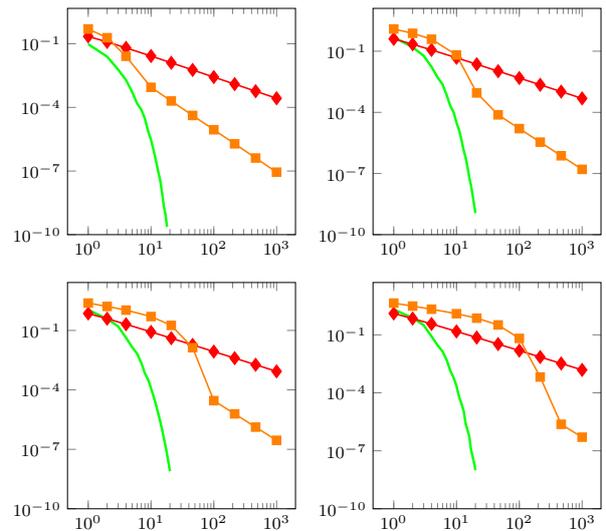

\subsection{Quality of compression}
In order to ensure that decoding the edge-compressed pictures by the Krylov method does not affect the quality of the recovered image, we provide experiments with pictures of the Kodak lossless true colour image suite (cf. Franzen~\cite{Kodak}).
In order to measure the deviation of the decoded compressed picture from the original picture, we use the mean-square error (MSE). For two colour pictures $\vec{u},\vec{v} \in \mathbb{R}^{M,N,3}$ with three colour channels and dimension $M \times N$, the mean-square error is given as
\begin{equation*}
        \operatorname{MSE}(\vec{u}, \vec{v}) \coloneqq \frac 1 {3MN}\sum_{k = 1}^3\sum_{i = 1}^M\sum_{j = 1}^N (u_{i,j,k} - v_{i,j,k})^2.
\end{equation*}
We use $\mu = 2$ in Algorithm~\ref{alg:mucycle}, which corresponds to the W-cycle, and $4$ pre- as well as post-relaxation steps. For the nested iteration in order to obtain a good starting value, we use $\nu_0=1$ in Algorithm~\ref{alg:nestiter}. For these pictures, we use $7$ levels in the multigrid method.

We found that for a large time $t=10^7$, the extended Krylov subspace with dimension $m=3$ is sufficient to provide a good reconstruction.
This means that only one solve of a linear system with the multigrid method is necessary.
For the optimal $\gamma=1.5\slash t$ and picture kodim07 of the test suite, the original picture and the reconstruction can be seen in Figure~\ref{kodim07} on the left-hand and right-hand side, respectively.
\begin{figure}[t]
        \centering
        \subfigure[original image]{\includegraphics[height=5cm]{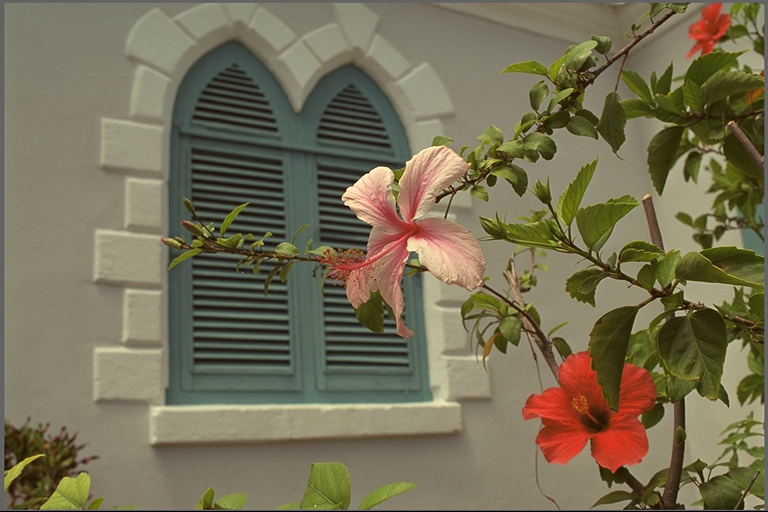}} \quad
        \subfigure[reconstruction, {\small $t=10^7$}]{\includegraphics[height=5cm]{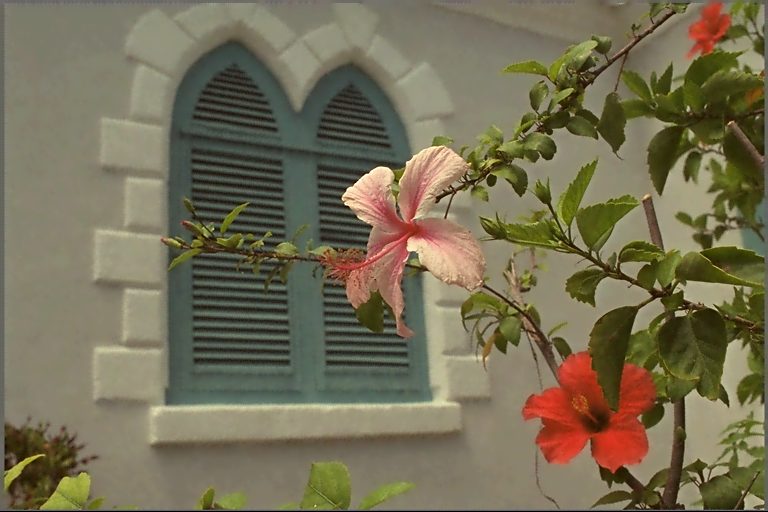}}
        \caption{(a) shows the original image and (b) shows the reconstruction by computing the solution of the heat equation up to time $t=10^7$.} \label{kodim07}
\end{figure}
The image has been compressed by the dithering-based method.
The same observation turned out to be true for the whole test set and for dithering-based as well as the edge-based compression of the images. We present the results in Table~\ref{tab:compress}. Here, we also state the peak signal-to-noise ratio (PSNR),
\begin{equation*}
\operatorname{PSNR}(\vec{u}, \vec{v}) \coloneqq 10 \cdot \log_{10} \left(\frac {255^2} {\operatorname{MSE}(\vec{u}, \vec{v})}\right) [\si{\dB}],
\end{equation*}
which is the most commonly used measure for the quality of reconstructions in lossy compression schemes.
We also state the compression rate in \emph{bits per pixels (bpp)} which refers to the average number of bits needed to encode each image pixel. The original pictures are RGB pictures using $8$ bits per colour channel which gives $24$ bpp in the original pictures.
Better values are marked in bold. The averages (avg) over all values are shown in the last row.
\begin{table}[tbhp]
  \caption{Comparison of the quality of the reconstruction for the proposed decoding scheme for dithering-based as well as edge-based compressed images of the test suite.}\label{tab:compress}
\begin{center}
\begin{tabular}{|l|ccc|ccc|l|ccc|ccc|}
\hline
\multirow{2}{*}{img} & \multicolumn{3}{c|}{dithering-based} & \multicolumn{3}{c|}{edge-based} &
\multirow{2}{*}{img} & \multicolumn{3}{c|}{dithering-based} & \multicolumn{3}{c|}{edge-based} \\
& bpp & MSE & PSNR  & bpp & MSE & PSNR &
& bpp & MSE & PSNR  & bpp & MSE & PSNR \\\hline
01 & 2.37 & \textbf{161.42} & \textbf{26.05} & \textbf{2.17} & 164.39 & 25.97 &
13 & \textbf{2.70} & \textbf{261.84} & \textbf{23.95} & 2.99 & 286.27 & 23.56 \rule{0pt}{2.5ex} \\
02 & 2.23 & \textbf{26.63} & \textbf{33.88} & \textbf{1.99} & 64.59 & 30.03 &
14 & 2.63 & \textbf{77.51} & \textbf{29.24} & \textbf{2.50} & 101.81 & 28.05 \\
03 & 2.18 & \textbf{14.23} & \textbf{36.60} & \textbf{1.66} & 50.11 & 31.13 &
15 & 2.34 & \textbf{26.98} & \textbf{33.82} & \textbf{1.86} & 74.44 & 29.41 \\
04 & 2.37 & \textbf{26.31} & \textbf{33.93} & \textbf{2.11} & 47.98 & 31.32 &
16 & 2.09 & \textbf{35.29} & \textbf{32.65} & \textbf{1.81} & 63.26 & 30.12 \\
05 & 2.74 & \textbf{144.67} & \textbf{26.53} & \textbf{2.43} & 165.88 & 25.93 &
17 & 2.28 & \textbf{22.83} & \textbf{34.55} & \textbf{1.92} & 53.52 & 30.85 \\
06 & 2.34 & \textbf{86.72} & \textbf{28.75} & \textbf{2.00} & 148.47 & 26.41 &
18 & 2.71 & \textbf{74.10} & \textbf{29.43} & \textbf{2.57} & 120.92 & 27.31 \\
07 & 2.38 & \textbf{22.65} & \textbf{34.58} & \textbf{1.45} & 62.19 & 30.19 &
19 & 2.30 & \textbf{52.76} & \textbf{30.91} & \textbf{1.70} & 106.74 & 27.85 \\
08 & 2.68 & \textbf{271.95} & \textbf{23.79} & \textbf{1.99} & 282.04 & 23.63 &
20 & 2.05 & \textbf{22.22} & \textbf{34.66} & \textbf{1.28} & 68.80 & 29.76 \\
09 & 2.16 & \textbf{20.77} & \textbf{34.96} & \textbf{1.19} & 51.70 & 31.00 &
21 & 2.36 & \textbf{49.72} & \textbf{31.17} & \textbf{1.72} & 111.10 & 27.67 \\
10 & 2.22 & \textbf{22.90} & \textbf{34.53} & \textbf{1.61} & 46.91 & 31.42 &
22 & 2.59 & \textbf{42.20} & \textbf{31.88} & \textbf{2.42} & 72.10 & 29.55 \\
11 & 2.41 & \textbf{53.89} & \textbf{30.82} & \textbf{2.06} & 86.52 & 28.76 &
23 & 2.36 & \textbf{9.86} & \textbf{38.19} & \textbf{1.88} & 39.84 & 32.13 \\
12 & 2.13 & \textbf{20.24} & \textbf{35.07} & \textbf{1.77} & 39.67 & 32.15 &
24 & 2.57 & \textbf{102.71} & \textbf{28.01} & \textbf{2.18} & 153.04 & 26.28 \\
\hline
\multicolumn{7}{c}{} & avg & 2.38 & \textbf{68.77} & \textbf{31.58} & \textbf{1.97} & 102.60 & 28.77 \rule{0pt}{2.5ex}\\
\end{tabular}
\end{center}
\end{table}
The results in Table~\ref{tab:compress} show that the decoding method is sufficiently accurate.

\begin{figure}[t]
        \centering
        \subfigure[compression] {\includegraphics[height=3.5cm]{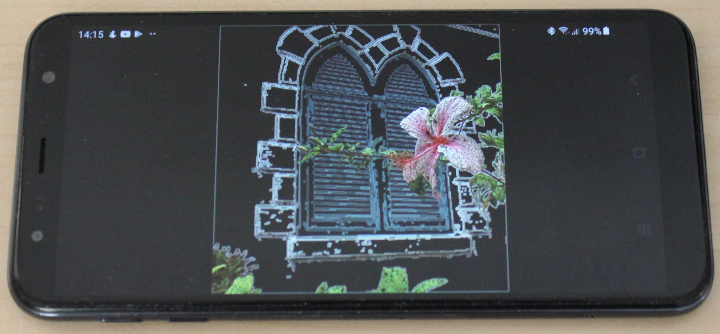}} \quad \qquad
        \subfigure[reconstruction] {\includegraphics[height=3.5cm]{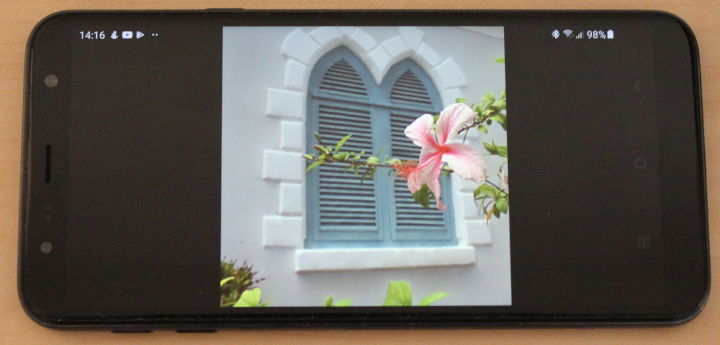}}
        \caption{Experiment on a smartphone} \label{phonetest}
\end{figure}

\subsection{Performance on an everyday device}

High compression rates are particularly important for embedded devices like smartphones, smart TV sets, and smart watches where storage is limited. Nowadays, these devices include embedded GPUs (Graphics Processing Unit) which allow to accelerate image processing tasks, considerably.
The industry standard to accelerate graphics by the use of these GPUs is OpenGL ES (Open Graphics Library for Embedded Systems) managed  by the non-profit Khronos group (cf. \cite{khronos}).
For our experiment, we will use the version OpenGL ES 3.2, which is available on $88.6\%$ of the devices running Android as of the 23rd April 2024 (cf. \cite{OpenGLdistro}) as well as the version OpenGL ES 3.1 with the extension GL\_EXT\_color\_buffer\_float which allows to render to float textures attached to a framebuffer. Our approach, with the implementation details given in Section~\ref{sec:multigrid}, perfectly fits to the OpenGL application interface. Pictures are treated as textures that are operated on in a parallel manner by the use of vertex and fragment shaders. Turning the matrices in sparse formats would not lead to algorithms that can be easily ported to embedded GPUs. We first used a desktop computer with a NVIDIA GeForce RTX 3060/PCIe/SSE2. OpenGL ES 3.2 is available on this GPU. The average time of ten runs of our program to decode
the $512 \times 512$ RGB picture  in the middle of Figure~\ref{overallidea} to the picture on the right-hand side of Figure~\ref{overallidea} was $0.014\,s$.
The dimension of system \eqref{ode} is $786432$ for this picture. On a notebook with the integrated graphics processor Intel(R) HD Graphics 620 (KBL GT2), the average time of ten runs was
$0.049\,s$. OpenGL ES 3.2 is also available on this graphics processor.
As an embedded device, we used a Samsung Galaxy J4+ smartphone running Android version 9 (pie) with a Qualcomm Adreno $308$ GPU with the same picture.
This phone allows for version OpenGL ES 3.1 with GL\_EXT\_color\_buffer\_float extension. The C code was compiled with the native development toolkit for Android systems (cf.~\cite{AndroidNDK}). The transition from the left-hand side of Figure~\ref{phonetest}
to the right-hand side of Figure~\ref{phonetest} took about $0.718\,s$ in average. With less than a second, this seems to be fast enough to decode edge-compressed pictures stored on this phone in a real-life application. We also tested our algorithm on a high-end smartphone
with a Qualcomm Adreno $740$ GPU running Android version 14. This phone supports OpenGL ES 3.2 and the decoding took $0.034\,s$ in average.
At that speed, touching the compression on the left-hand side of Figure~\ref{phonetest} immediately turns the image to the reconstruction on the right-hand side of Figure~\ref{phonetest}. The experience is smooth and one does not even recognize any slightest delay.


\section{Conclusion} \label{sec:conclusion}
We presented an efficient method to solve inpainting problems by homogeneous diffusion based on extended Krylov subspaces. The method is basically applicable to all inpainting problems of this type. We studied the problem of decoding edge-based and dithering-based compressed images, where the boundaries for the inpainting problems are especially challenging. To our best knowledge, no other method is known that can provably solve the heat equation up to a prescribed large time $t$ with such an accuracy and efficiency. 


\begin{thebibliography}{99}

\bibitem{Carlsson88}
Carlsson S.
\newblock Sketch based coding of grey level images.
\newblock Signal Processing. 1988;{\bf 15}(1):57--83.

\bibitem{Elder99}
Elder JH.
\newblock Are edges incomplete?
\newblock International Journal of Computer Vision. 1999;{\bf 34}:97--122.

\bibitem{HuMo89}
Hummel R, and Moniot R.
\newblock Reconstructions from zero-crossings in scale space.
\newblock IEEE Transactions on Acoustics, Speech and Signal Processing.
  1989;{\bf 37}:2111--2130.

\bibitem{ZeRot86}
Zeevi Y, and Rotem D.
\newblock Image reconstruction from zero-crossings.
\newblock IEEE Transactions on Acoustics, Speech and Signal Processing.
  1986;{\bf 34}:1269--1277.

\bibitem{Reidetal97}
Reid MM, Millar RJ, and Black ND.
\newblock {Second-generation image coding: an overview}.
\newblock ACM Computing Surveys. 1997;{\bf 29}:3--29.

\bibitem{Mainbergeretal11}
Mainberger M, Bruhn A, Weickert J, and Forchhammer S.
\newblock Edge-based compression of cartoon-like images with homogeneous
  diffusion.
\newblock Pattern Recognition. 2011;{\bf 44}(9):1859--1873.

\bibitem{Hoeltgenetal13}
Hoeltgen L, Setzer S, and Weickert J.
\newblock An Optimal Control Approach to Find Sparse Data for Laplace
  Interpolation.
\newblock In: Heyden A, Kahl F, Olsson C, Oskarsson M, and Tai XC, editors.
  Energy Minimization Methods in Computer Vision and Pattern Recognition.
  Lecture notes in Computer Science. vol. 8081. Berlin: Springer; 2013. p.
  151--164.

\bibitem{Mainbergeretal12}
Mainberger M, Hoffmann S, Weickert J, Tang CH, Johannsen D, Neumann F, et~al.
\newblock Optimising Spatial and Tonal Data for Homogeneous Diffusion
  Inpainting.
\newblock In: Bruckstein AM, ter Haar~Romney BM, Bronstein AM, and Bronstein
  MM, editors. Scale Space and Variational Methods in Computer Vision. SSVM
  2011. Lecture Notes in Computer Science. vol. 6667. Berlin, Heidelberg:
  Springer; 2012. p. 27--37.

\bibitem{Grietal06}
Grimm V, Henn S, and Witsch K.
\newblock A higher-order {PDE}-based image registration approach.
\newblock Numer Linear Algebra Appl. 2006;{\bf 13}(5):399--417.

\bibitem{hoacta10}
Hochbruck M, and Ostermann A.
\newblock Exponential integrators.
\newblock Acta Numer. 2010;{\bf 19}:209--286.

\bibitem{inpaintingcarola15}
Sch\"{o}nlieb CB.
\newblock Partial differential equation methods for image inpainting. vol.~29
  of Cambridge Monographs on Applied and Computational Mathematics.
\newblock Cambridge University Press, New York; 2015.

\bibitem{Ruhe84}
Ruhe A.
\newblock Rational {K}rylov sequence methods for eigenvalue computation.
\newblock Linear Algebra Appl. 1984;{\bf 58}:391--405.

\bibitem{Ruhe98}
Ruhe A.
\newblock Rational {K}rylov: A practical algorithm for large sparse nonsymetric
  matrix pencils.
\newblock SIAM J Sci Comput. 1998;{\bf 19}:1535--1551.

\bibitem{BreNoRe12}
Brezinski C, Novati P, and Redivo-Zaglia M.
\newblock A rational {A}rnoldi approach for ill-conditioned linear systems.
\newblock J Comput Appl Math. 2012;{\bf 236}(8):2063--2077.

\bibitem{BuDoRei17}
Buccini A, Donatelli M, and Reichel L.
\newblock Iterated {T}ikhonov regularization with a general penalty term.
\newblock Numer Linear Algebra Appl. 2017;{\bf 24}(4):e2089, 12.

\bibitem{invkry}
Grimm V.
\newblock A conjugate-gradient-type rational {K}rylov subspace method for
  ill-posed problems.
\newblock Inverse Problems. 2020;{\bf 36}(1):015008, 19.

\bibitem{RamRei19}
Ramlau R, and Reichel L.
\newblock Error estimates for {A}rnoldi-{T}ikhonov regularization for ill-posed
  operator equations.
\newblock Inverse Problems. 2019;{\bf 35}(5):055002, 23.

\bibitem{GG13}
G{\"o}ckler T, and Grimm V.
\newblock Convergence {A}nalysis of an {E}xtended {K}rylov {S}ubspace {M}ethod
  for the {A}pproximation of {O}perator {F}unctions in {E}xponential
  {I}ntegrators.
\newblock SIAM J Numer Anal. 2013;{\bf 51}(4):2189--2213.

\bibitem{GG14}
G{\"o}ckler T, and Grimm V.
\newblock Uniform approximation of {$\varphi$}-functions in exponential
  integrators by a rational {K}rylov subspace method with simple poles.
\newblock SIAM J Matrix Anal Appl. 2014;{\bf 35}(4):1467--1489.

\bibitem{ratkryphi11}
Grimm V.
\newblock Resolvent {K}rylov subspace approximation to operator functions.
\newblock BIT Numerical Mathematics. 2012;{\bf 52}(3):639--659.

\bibitem{GGautosmooth17}
Grimm V., and G\"ockler T.
\newblock Automatic smoothness detection of the resolvent {K}rylov subspace
  method for the approximation of {$C_0$}-semigroups.
\newblock SIAM J Numer Anal. 2017;{\bf 55}(3):1483--1504.

\bibitem{And81}
Andersson JE.
\newblock {Approximation of $e^{-x}$ by rational functions with concentrated
  negative poles}.
\newblock J~Approx~Theory. 1981;{\bf 32}(2):85--95.

\bibitem{marlis_jasper}
van~den Eshof J, and Hochbruck M.
\newblock Preconditioning {L}anczos approximations to the matrix exponential.
\newblock SIAM J Sci Comp. 2006;{\bf 27}(4):1438--1457.

\bibitem{beckermann_guettel12}
Beckermann B, and G{\"u}ttel S.
\newblock Superlinear convergence of the rational {A}rnoldi method for the
  approximation of matrix functions.
\newblock Numer Math. 2012;{\bf 121}(2):205--236.

\bibitem{Druskin_Knizhnerman98}
Druskin V, and Knizhnerman L.
\newblock Extended {K}rylov subspaces: approximation of the matrix square root
  and related functions.
\newblock SIAM J Matrix Anal Appl. 1998;{\bf 19}(3):755--771.

\bibitem{KnizhSimoncini09}
Knizhnerman L, and Simoncini V.
\newblock {A new investigation of the extended Krylov subspace method for
  matrix function evaluations}.
\newblock Numer Linear Algebra Appl. 2010;{\bf 17}(4):615--638.

\bibitem{AlMohyHigham11}
Al-Mohy AH, and Higham NJ.
\newblock Computing the action of the matrix exponential, with an application
  to exponential integrators.
\newblock SIAM J Sci Comput. 2011;{\bf 33}(2):488--511.

\bibitem{highambook}
Higham NJ.
\newblock Functions of matrices.
\newblock Society for Industrial and Applied Mathematics (SIAM), Philadelphia,
  PA; 2008.

\bibitem{Kodak}
Franzen R. {Kodak Lossless True Color Image Suite}.
\newblock Accessed: 26 February 2024;.
\newblock \url{http://r0k.us/graphics/kodak}.

\bibitem{Marr76}
Marr D.
\newblock Early Processing of Visual Information.
\newblock Philosophical Transactions of the Royal Society of London Series B,
  Biological Sciences. 1975;{\bf 275}(942):483--519.

\bibitem{MarrHildreth80}
Marr D, and Hildreth E.
\newblock {Theory of edge detection}.
\newblock Proceedings of the Royal Society of London. 1980;{\bf B
  207}:187--217.

\bibitem{Canny86}
Canny J.
\newblock {A Computational Approach of Edge Detection}.
\newblock IEEE Transactions on Pattern Analysis and Machine Intelligence.
  1986;{\bf PAMI-8}(6):679--698.

\bibitem{Belhachmietal09}
Belhachmi Z, Bucur D, Burgeth B, and Weickert J.
\newblock How to choose interpolation data in images.
\newblock SIAM J Appl Math. 2009;{\bf 70}(1):333--352.

\bibitem{FloStein76}
Floyd RW, and Steinberg L.
\newblock {An adaptive algorithm for spatial grey scale}.
\newblock Proceedings of the Society of Information Display. 1976;{\bf
  17}:75--77.

  \bibitem{Saaditer}
Saad Y.
\newblock Iterative methods for sparse linear systems.
\newblock 2nd ed. Society for Industrial and Applied Mathematics, Philadelphia,
  PA; 2003.

\bibitem{Bruhnetal05}
Bruhn A, Weickert J, Feddern C, Kohlberger T, and Schn\"{o}rr C.
\newblock Variational optic flow computation in real-time.
\newblock IEEE Transactions on Image Processing. 2005;{\bf 14}(5):608--615.

\bibitem{khronos}
Khronos Group.
\newblock \url{https://www.khronos.org/}; 2021, (accessed 21 February 2025).

\bibitem{OpenGLdistro}
Distribution dashbord.
\newblock
  \url{https://developer.android.com/about/dashboards/index.html#OpenGL};
  accessed 21 February 2025.

\bibitem{AndroidNDK}
Android NDK.
\newblock \url{https://developer.android.com/ndk}; accessed 21 February 2025.

\end{thebibliography}
\end{document}